\documentclass[a4paper,USenglish,cleveref,autoref,thm-restate,nolineno]{socg-lipics-v2021}

\pdfoutput=1 %uncomment to ensure pdflatex processing (mandatatory e.g. to submit to arXiv)
\hideLIPIcs  %uncomment to remove references to LIPIcs series (logo, DOI, ...), e.g. when preparing a pre-final version to be uploaded to arXiv or another public repository

\usepackage{amsmath, amssymb, amsthm, amscd}
\usepackage{url}
\usepackage{graphicx}
\usepackage{stmaryrd}
\usepackage[ruled, vlined]{algorithm2e}
\usepackage{tikz}
\usetikzlibrary{calc}
\usetikzlibrary{arrows}
\usetikzlibrary{decorations.markings}
\usetikzlibrary{decorations.pathreplacing}
\usetikzlibrary{arrows,shapes,positioning}
\usetikzlibrary{patterns}
\tikzstyle directed=[postaction={decorate,decoration={markings,
    mark=at position #1 with {\arrow{>}}}}]
\tikzstyle rdirected=[postaction={decorate,decoration={markings,
    mark=at position #1 with {\arrow{<}}}}]
\tikzset{anchorbase/.style={baseline={([yshift=-0.5ex]current bounding box.center)}}}

\tikzset{
    partial ellipse/.style args={#1:#2:#3}{
        insert path={+ (#1:#3) arc (#1:#2:#3)}
    }
  }

\renewcommand{\hat}{\widehat}

\newcommand{\symG}{\Sigma}%{\mathfrak{S}}%{\Sigma}
\newcommand{\Perm}{\symG}

\newcommand{\ie}{{\em i.e.}}
\newcommand{\eg}{{\em e.g.}}

\newcommand{\He}{\mathcal{H}}
\newcommand{\Br}{B}
\newcommand{\Z}{\mathbb{Z}}

\newcommand{\Tr}{\mathrm{Tr}}

\newcommand{\idx}{\operatorname{idx}}
\newcommand{\projlength}{\mathrm{projlength}}
\newcommand{\val}{\mathrm{val}}
\newcommand{\R}{\mathbb{R}}
\newcommand{\Hecke}{\mathcal{H}}
\newcommand{\poly}{\operatorname{poly}}
\newcommand{\tw}{\operatorname{tw}}
\newcommand{\hecke}{\mathtt{hecke}}
\newcommand{\id}{\operatorname{id}}

\title{A fast algorithm for the Hecke representation of the braid group, and applications to the computation of the HOMFLY-PT polynomial and the search for interesting braids}

\titlerunning{A fast algorithm for the Hecke representation, and applications}

\author{Clément Maria}{Inria d'Université Côte d'Azur, France \and \url{https://www-sop.inria.fr/members/Clement.Maria/} }{ clement.maria@inria.fr}{https://orcid.org/0000-0002-2007-2584}{Partially supported by the ANR project ANR-20-CE48-0007 (AlgoKnot).}

\author{Hoel Queffelec}{France-Australia Mathematical Sciences and Interactions ANU - CNRS International Research Laboratory, Australia \and \url{https://imag.umontpellier.fr/~queffelec/}}{hoel.queffelec@cnrs.fr}{}{}

\authorrunning{C. Maria and H. Queffelec} %TODO mandatory. First: Use abbreviated first/middle names. Second (only in severe cases): Use first author plus 'et al.'

\Copyright{Clément Maria and Hoel Queffelec} %TODO mandatory, please use full first names. LIPIcs license is "CC-BY";  http://creativecommons.org/licenses/by/3.0/


 %TODO mandatory; please add comma-separated list of keywords

\category{} %optional, e.g. invited paper

\relatedversion{} %optional, e.g. full version hosted on arXiv, HAL, or other respository/website
\acknowledgements{We would like to thank Asilata Bapat, Stepan Orevkov---who pointed us to Morton's 1985 implementation~\cite{Morton_code} of the HOMFLY-PT algorithm---, and Oded Yacobi for helpful discussions. Some experiments presented in this paper were carried out using the Grid'5000 testbed, supported by a scientific interest group hosted by Inria and including CNRS, RENATER and several Universities as well as other organizations (see \url{https://www.grid5000.fr}).}%optional

\nolinenumbers %uncomment to disable line numbering

\begin{document}

\maketitle

%TODO mandatory: add short abstract of the document
\begin{abstract}
Knot theory is an active field of mathematics, in which combinatorial and computational methods play an important role. One side of computational knot theory, that has gained interest in recent years, both for complexity analysis and practical algorithms, is quantum topology and the computation of topological invariants issued from the theory.  

In this article, we leverage the rigidity brought by the representation-theoretic origins of the quantum invariants for algorithmic purposes. We do so by exploiting braids and the algebraic properties of the braid group to describe, analyze, and implement a fast algorithm to compute the Hecke representation of the braid group. We apply this construction to design a parameterized algorithm to compute the HOMFLY-PT polynomial of knots, and demonstrate its interest experimentally. Finally, we combine our fast Hecke representation algorithm with Garside theory, to implement a reservoir sampling search and find non-trivial braids with trivial Hecke representations with coefficients in $\mathbb{Z}/p\mathbb{Z}$. We find several such braids, in particular proving that the Hecke representation of $B_5$ with $\mathbb{Z}/2\mathbb{Z}$ coefficients is non-faithful, a previously unknown fact.  
\end{abstract}

\section{Introduction}
Geometrically, a {\em braid} on $n$ strands is the embedding of $n$ non-intersecting paths in the 3-dimensional space $\R^2 \times [0,1]$, such that every path connects a point $(i,0,0), i \in \{1,\ldots,n\}$ in the bottom plane to a point $(j,0,1), j \in \{1,\ldots,n\}$ in the top plane, and every path grows monotonically along the z-axis. Two braids are {\em equivalent} if there is an ambient isotopy of $\R^2 \times [0,1]$ fixing the bottom and top planes and taking one braid to the other. Braids are generally represented by {\em braid diagrams}, that are planar projections of a braid along the y-axis, keeping track of upper and under crossings; see Figure~\ref{fig:intro_fig}.

Braids are notably important in knot theory, as any link can be represented as the closure of a braid~\cite{Alexander1923}. Knots have been studied extensively under the algorithmic lens. A famous problem is the computational complexity of recognizing the {\em trivial knot} from an input diagram, which is known to be in the complexity classes ${\bf NP}$~\cite{Hass1999} and ${\bf coNP}$~\cite{Lackenby2021}, for which the best known worst case algorithm is exponential~\cite{Hass1999}, but which experimentally exhibits a fast polynomial time behavior~\cite{burton2014fastbranchingalgorithmunknot} with optimized implementation. In particular, the experimental aspects of computational knot theory play a fundamental role in the field, where mathematicians and computer scientists use efficient software, such as {\tt Regina}~\cite{burton04-regina,regina} and {\tt SnapPy}~\cite{snappy}, as well as computer-constructed census of knots~\cite{burton:LIPIcs.SoCG.2020.25}, to guess and challenge profound conjectures, \eg\cite{BH,AIF_0__0_0_A159_0,Garoufalidis2005}. Consequently, an important side of computational knot theory is the design and implementation of fast, highly optimized algorithms.

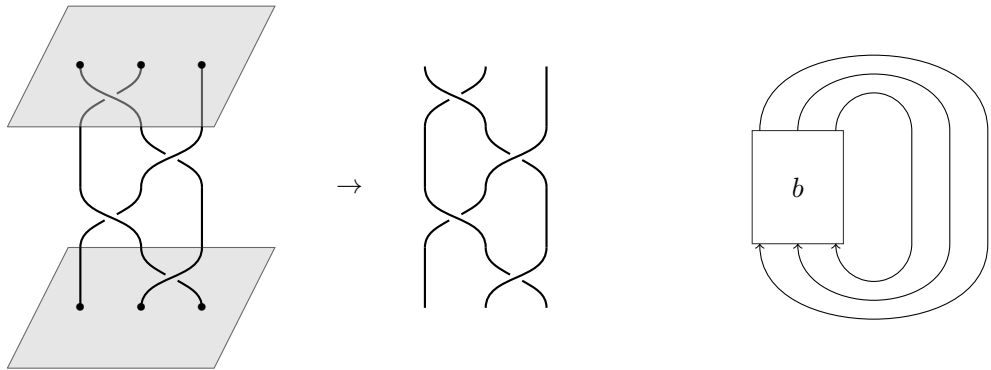
\begin{figure}[t!]
  \centering
  \begin{subfigure}[b]{0.69\textwidth}
    \centering
    \[
    \begin{tikzpicture}[anchorbase,scale=0.8]
      %% Bottom plane
      \draw [opacity=.6,fill opacity=.1,fill=black] (-1.2,-1) -- (2.2,-1) -- (3.2,1) -- (-.2,1) -- (-1.2,-1);
      \node at (0,0) {\textbullet};
      \node at (1,0) {\textbullet};
      \node at (2,0) {\textbullet};
      %% Braid
      \draw [thick] (0,0) -- (0,1);
      \draw [thick] (1,0) to [out=90,in=-90] (2,1);
      \draw [thick] (2,0) to [out=90,in=-30] (1.6,.45);
      \draw [thick] (1.4,.55) to [out=150,in=-90] (1,1);
      \draw[thick] (0,1) to [out=90,in=-150] (.4,1.45);
      \draw [thick] (.6,1.55) to [out=30,in=-90] (1,2);
      \draw [thick] (1,1) to [out=90,in=-90] (0,2);
      \draw [thick] (2,1) -- (2,2);
      \draw [thick] (0,2) -- (0,3);
      \draw [thick] (1,2) to [out=90,in=-90] (2,3);
      \draw [thick] (2,2) to [out=90,in=-30] (1.6,2.45);
      \draw [thick] (1.4,2.55) to [out=150,in=-90] (1,3);
      \draw[thick,opacity=.6] (0,3) to [out=90,in=-150] (.4,3.45);
      \draw [thick,opacity=.6] (.6,3.55) to [out=30,in=-90] (1,4);
      \draw [thick,opacity=.6] (1,3) to [out=90,in=-90] (0,4);
      \draw [thick,opacity=.6] (2,3) -- (2,4);
      %% Top plane
      \draw [opacity=.6,fill opacity=.1,fill=black] (-1.2,3) -- (2.2,3) -- (3.2,5) -- (-.2,5) -- (-1.2,3);
      \node at (0,4) {\textbullet};
      \node at (1,4) {\textbullet};
      \node at (2,4) {\textbullet};
    \end{tikzpicture}
    \qquad \rightarrow \qquad
    \begin{tikzpicture}[anchorbase,scale=0.8]
      %% Braid
      \draw [thick] (0,0) -- (0,1);
      \draw [thick] (1,0) to [out=90,in=-90] (2,1);
      \draw [thick] (2,0) to [out=90,in=-30] (1.6,.45);
      \draw [thick] (1.4,.55) to [out=150,in=-90] (1,1);
      \draw[thick] (0,1) to [out=90,in=-150] (.4,1.45);
      \draw [thick] (.6,1.55) to [out=30,in=-90] (1,2);
      \draw [thick] (1,1) to [out=90,in=-90] (0,2);
      \draw [thick] (2,1) -- (2,2);
      \draw [thick] (0,2) -- (0,3);
      \draw [thick] (1,2) to [out=90,in=-90] (2,3);
      \draw [thick] (2,2) to [out=90,in=-30] (1.6,2.45);
      \draw [thick] (1.4,2.55) to [out=150,in=-90] (1,3);
      \draw[thick] (0,3) to [out=90,in=-150] (.4,3.45);
      \draw [thick] (.6,3.55) to [out=30,in=-90] (1,4);
      \draw [thick] (1,3) to [out=90,in=-90] (0,4);
      \draw [thick] (2,3) -- (2,4);
    \end{tikzpicture}
    \]
  \end{subfigure}\hfill
  \begin{subfigure}[b]{0.3\textwidth}
    \centering
    \[
    \begin{tikzpicture}[anchorbase,scale=.5]
      %% Square
      \draw [opacity=.8] (-.2,0) rectangle (2.2,3);
      \node at (1,1.5) {\large $b$};
      \draw [->] (0,3) to [out=90,in=180] (3,5) to [out=0,in=90] (6,3) -- (6,0) to [out=-90,in=0] (3,-2) to [out=180,in=-90] (0,0);
      \draw [->] (1,3) to [out=90,in=180] (3,4.5) to [out=0,in=90] (5,3) -- (5,0) to [out=-90,in=0] (3,-1.5) to [out=180,in=-90] (1,0);
      \draw [->] (2,3) to [out=90,in=180] (3,4) to [out=0,in=90] (4,3) -- (4,0) to [out=-90,in=0] (3,-1) to [out=180,in=-90] (2,0);
    \end{tikzpicture}
    \]
  \end{subfigure}
\caption{Left: Geometric braid on $3$ strands in $\R^2 \times [0;1]$, and natural diagram (projection on $xz$-plane). Right: Braid closure $\hat{b}$ of braid $b$.}
\label{fig:intro_fig}
\end{figure}

Motivated by a finer understanding of the complexity of problems related to knots as well as practically fast computation, a recent route of research uses tools from {\em parameterized complexity} to compute {\em topological invariants} of knots. This approach has been particularly successful for invariants constructed via {\em quantum topology}, a field of topology using tools from quantum mechanics. In particular, the {\em treewidth} of a graph is a parameter measuring how close the graph is to a tree and, similarly, the {\em pathwidth} measures the proximity to a path. They can be extended to knot theory by considering the graph obtained from a diagram by putting a vertex on each crossing and an edge for each strand connecting crossings.

These are important parameters capturing a certain notion of sparsity of the input, and they can be combined with algorithmic techniques such as {\em dynamic programming} to design algorithms whose complexity depends exponentially on the tree/pathwidth and only polynomially in the size of the input: in quantum topology, such algorithms have been designed for the Jones and Kauffman polynomials~\cite{MAKOWSKY2003742}, the Reshetikhin-Turaev invariants~\cite{maria:LIPIcs.SoCG.2021.53}, and the HOMFLY-PT polynomial~\cite{burton:LIPIcs.SoCG.2018.18}, which are {\bf \#P-hard} to compute in general~\cite{Kuperberg2015}. Note that these quantum invariants, together with theoretically and practically fast algorithms, have been applied to construct knot censuses~\cite{burton:LIPIcs.SoCG.2020.25}. 

Contrary to knots, braids on $n$-strands have a natural algebraic description as a group, the {\em braid group} $B_n$, yielding rich algebraic properties. To our knowledge, braids have not received the same attention from the computer science community compared to knots.

\subparagraph*{Our results.} The goal of this article is to describe fast algorithms and data structures for braids and the braid group, related to quantum topology, and to apply them to computational knot theory and experimental mathematics. As opposed to algorithms on knots mentioned above, these algorithms rely heavily on the algebraic structure of the braid group. 

Our starting point is to consider the {\em Hecke representation} of a braid. The {\em Hecke algebra} $\Hecke_n$ is a fundamental concept in modern mathematics~\cite{Mathas_book}, deeply connected to group theory, number theory, and knot theory. The braid group $B_n$ admits a representation into the Hecke algebra $\Hecke_n$, \ie, a map $B_n \to \Hecke$ respecting the group structure. Knowing whether this map (with $\Z$-coefficients) is \emph{faithful} (\ie, only the trivial braid has trivial image) is a major open question~\cite{Squier,Jones_Hecke,LX,Bigelow_IH}, related to the detection of the unknot by the Jones polynomial~\cite{Ito}.

Our algorithm consists of scanning an input braid diagram from bottom to top, and updating the representation of the braid in a basis of $\Hecke$. This natural strategy of computation has already been used by Hugh Morton in 1985~\cite{Morton_code}; the program that results from our work can be seen as a modern and optimized implementation of Morton's high level ideas.

\begin{theorem}[see Theorem~\ref{thm:complexity_hecke}]\label{thm:main_hecke}
Given a braid $b \in B_n$ on $n$ strands and with $N$ crossings, there is an algorithm to compute its Hecke representation in $O(n! (n\log n) \times N)$ operations and $O(n! \times N)$ algebraic operations, storing $n! + O(1)$ algebraic elements. 
\end{theorem}

Every knot or link can be obtained from the braid closure of a braid~\cite{Alexander1923}; see Figure~\ref{fig:intro_fig}. Additionally, the HOMFLY-PT polynomial of the knot/link obtained from the braid closure can be computed from the Hecke representation of the braid by taking a {\em trace}. This is an alternative, more algebraic approach, to the usual definition of quantum invariants used in parameterized algorithms in the literature, relying either on {\em state sums}~\cite{Burton2018}, {\em tensor networks}~\cite{maria:LIPIcs.SoCG.2021.53}, or skein relations~\cite{burton:LIPIcs.SoCG.2018.18}. Here, we exploit the rigidity provided by the algebraic or representation-theoretic point of view allows to more efficient algorithmic processes. In particular, the morphism spaces between representations are of much lower dimension than those between the underlying vector spaces, which helps reducing the algorithmic complexity.

We describe a fast implementation of the trace operation in the single pass of the coordinate vector representing the Hecke element associated to the braid.

\begin{theorem}[see Theorem~\ref{thm:complexity_homfltypt}]
Given a braid $b \in B_n$ with $m$ crossings, there is an algorithm to compute its HOMFLY-PT polynomial with the same complexity as Theorem~\ref{thm:main_hecke}.
\end{theorem}

We compare experimentally this algorithm against Burton's implementation~\cite{burton:LIPIcs.SoCG.2018.18,regina} on large families of braid closures, and demonstrate its practical interest.

Finally, we use these fast implementations in experimental mathematics, running an extensive search for counter-examples to the Hecke faithfulness question in the case of Hecke representation with $\Z/p\Z$ coefficients. 

Taking inspiration from earlier works~\cite{BQ,Bigelow_Burau,GWY}, we have implemented a random algorithm for finding counter-examples to faithfulness in the braid groups $B_4$ and $B_5$, for different coefficients $\Z/p\Z$. Exploiting the favorable algebraic properties of braids, we have used {\em Garside theory} to generate increasingly complicated, non-trivial, braids for the search, running on each of them the algorithm from Theorem~\ref{thm:complexity_hecke} in order to find braids whose Hecke representation was getting increasingly close to the trivial one. We have found explicit non-trivial braids in $B_4$ whose Hecke image is trivial in coefficients $\Z/2\Z$, $\Z/3\Z$ and $\Z/4\Z$. The non-faithfulness of the Hecke representation was known abstractly in these cases, but our algorithm provides an efficient way to find examples of elements in the kernel. Additionally, we have found a non-trivial braid in $B_5$ whose Hecke image modulo 2 is trivial, while the status of the faithfulness of the Hecke representation of $B_5$ was unknown (in any coefficients).  

\begin{theorem}[see Theorem~\ref{thm:kernelB5}]
The Hecke representation of the braid group $B_5$ modulo $2$ is not faithful.
\end{theorem}

\subparagraph*{Comparison with the literature.} The HOMFLY-PT polynomial is a powerful topological invariants, and its computation has attracted the attention of mathematicians. For a knot diagram with $N$ crossings, Kauffman gave a fast $O(2^N)$ skein template algorithm~\cite{Kauffman1990}. More recently, starting from Kauffman's work, Burton~\cite{burton:LIPIcs.SoCG.2018.18} designed the first fixed parameter tractable algorithm in the treewidth of the knot diagram, running in $O((2 \tw)!^4 \cdot \poly(N))$,
and implemented it in the software {\tt Regina}~\cite{regina}. This complexity bound should be compared to our bound from Theorem~\ref{thm:main_hecke}.

While treewidth, over all possible diagrams of a knot, is a generally smaller parameter than the braid index, it is common for them to be essentially equal, and the diagrams minimizing the treewidth are braid closures ; this is the case for {\em torus knots} for example~\cite{DBLP:journals/jocg/SchleimerMPS19,lunel_et_al:LIPIcs.SoCG.2023.50}. On the other hand, the complexity dependence of our algorithm in the number of strands is much lower than Burton's dependence in the treewidth, which proves an advantage in practice on particular families of examples where treewidth and braid index are close.

Our approach to find counter-examples to the faithfulness of the Hecke representation is inspired from recent works~\cite{GWY} in computational mathematics, which have used a combination of curve-based random search with Garside theory to study the faithfulness of the {\em Burau representation} of the braid group, and in particular to find braids with trivial representations. 
However, the Burau representation of a braid is significantly less costly to compute, with a (small) polynomial dependence in both number of strands and number of crossings, while all known algorithms for the HOMFLY-PT polynomial, including this work, have a super-exponential dependence in the number of strands.

\section{Preliminaries}
\subsection{Permutations and compact encoding} \label{sec:perm}

The {\em symmetric group} $\symG_n$ is the group of permutations of $n$ elements $\{1, \ldots, n\}$. The group $\symG_n$ is generated by the {\em transpositions}, \ie, the permutations $s_i$ of the two consecutive elements at index $i$ and $i+1$, for $1 \leq i \leq n-1$. Every permutation $w \in \symG_n$ can be represented either by a unique {\em one-line word} $w(1)w(2) \cdots w(n)$, where $w(i)$ is the image of $i$ under the permutation $w$, or by a non-unique product of generators $w = s_{i_k} \cdots s_{i_1}$. 

An {\em inversion} in a permutation $w$ is a pair of indices $(i,j)$ such that $i<j$ and $w(i)>w(j)$. The {\em (Coxeter) length} of a permutation $w$, denoted by $\ell(w)$, is its total number of inversions.

The {\em Lehmer code} $L(w)$ of a permutation $w = w(1)w(2) \cdots w(n)$ is the sequence of positive integers $L(w) := (L(w)_1, \ldots, L(w)_n)$ such that:
\[
  L(w)_i := \#\{ j > i : w(j) < w(i) \},  
\]
\ie, $L(w)_i$ counts the number of inversions happening to the right of index $i$. In particular, $0 \leq L(w)_i \leq n-i$. Permutations are in bijection with their Lehmer code. 

This last properties allows us to represent a permutation with a single number, using the {\em factorial number system}. Specifically, given a sequence of numbers ${\bf a} = (a_1, \ldots, a_n)$ such that $0 \leq a_i \leq n-i$ for all $i$, one can write uniquely the sequence as a single positive integer:
\begin{equation}\label{eq:fact_system}
  (a_1, \ldots, a_n) \rightarrow F_{{\bf a}} = \sum_{i=1}^n a_i \times (n-i)!
\end{equation}

In particular, the application $\symG \to \{0, \ldots, n!-1\}$ sending $ w \mapsto F_{L(w)}$ is a bijection. In the following, we call the number $F_{L(w)}$ the {\em index of the permutation} $w$.

We prove in the appendix the following basic facts about the encoding:

\begin{lemma}
\label{lem:basic_perm}
\begin{enumerate}[(i)]
  \item Given an index $F_{L(w)} \in \{0, \ldots, n!-1\}$, there is an $O(n \log n)$ time algorithm to compute a one-line word presentation of $w$,
  \item Given a permutation $w$ represented as a one-line word $w(1) \cdots w(n)$, there is an $O(n \log n)$ time algorithm to compute its index $F_{L(w)}$,
  \item Let $w,w' \in \Perm_n$ be permutation, then:
  \[
    w(1) \cdots w(n) \leq_{\operatorname{lex}} w'(1) \cdots w'(n) \ \text{iff} \ F_{L(w)} \leq F_{L(w')},
  \]
  where $\leq_{\operatorname{lex}}$ denotes the lexicographic order.
\end{enumerate}
\end{lemma}

\subsection{Braids and Garside structure}

We consider $\Br_n$ the braid group on $n$ strands:
\[
\Br_n=\left\langle \sigma_1,\dots,\sigma_{n-1}\;|\; \begin{cases} \sigma_i\sigma_{i+1}\sigma_i=\sigma_{i+1}\sigma_i\sigma_{i+1}\; \\ \sigma_i\sigma_j=\sigma_j\sigma_i\;\text{if}\;|i-j|>1\end{cases}\right\rangle.
\]
The generators $\{\sigma_i^{\pm 1}\}_i$ are called the {\em Artin generators} of the braid group. The monoid with the same presentation, \ie, generated by the $\{\sigma_i^{+ 1}\}_i$, is denoted by $\Br_n^+$.

For some of our applications, we will be interested in producing braids that we want to make sure get more and more complicated. The notion of Garside structure~\cite{Dehornoy_Garside_book} will be instrumental, by providing us with an algorithmic-friendly normal form.

Let us denote $[1,\Delta]$ the braids that realize positive lifts of reduced words from the symmetric group. For example, $s_1s_2s_1=s_2s_1s_2\in \symG_n$ lifts to $\sigma_1\sigma_2\sigma_1=\sigma_2\sigma_1\sigma_2\in \Br_n$. It is a classical result that the map is well-defined (independent on the choice for the starting word). The half-twist $\Delta=\sigma_1\cdots\sigma_{n-1}\sigma_1\cdots \sigma_{n-2}\cdots\sigma_1$ is the so-called {\em Garside element}. Then one has a lattice structure~\cite{Garside,Dehornoy_Garside_book} on $[1,\Delta]$, extending to $\Br_n^+$, with left divisibility given by:
\[
b_1\leq_L b_2 \; \Leftrightarrow \; \exists b_3\in \Br_n^+,\; b_1b_3=b_2.
\]
Similarly, one can define a notion of right divisibility, $\leq_R$. We are now ready to obtain the normal form $b_k\cdots b_1$ of a braid $b\in\Br_n^+$ by inductively defining:
\[
b_i\colon=\sup\{b'\in [1,\Delta],\; b'\leq_R (bb_1^{-1}\cdots b_{i-1}^{-1})\}.
\]
If one adds the fact that any braid $b\in \Br_n$ can be written as $\Delta^rb^{+}$ with $b^+\in \Br_n^+$, then one has a preferred expression for any braid.

At some point, we will be interested in using the Garside normal form to generate complicated braids. This will be made algorithmic thanks to the following lemma.

\begin{lemma} \label{lem:NF_cond_back}
  The word $b_k\cdots b_1$ with $b_i\in [1,\Delta]$ is under normal form if and only if the following condition holds:
  \[\forall i\geq 2, \forall j,\quad \sigma_j\leq_R b_i\Rightarrow \sigma_j\leq_L b_{i-1}.\]
\end{lemma}
The divisibility in $[1,\Delta]$ being controlled by $\symG_n$, the computations can be made in the symmetric group. More details about the underlying automaton will appear in Section~\ref{sec:app_gar}.

\subsection{Hecke algebra}

The Hecke algebra is an ubiquitous object in representation theory. We introduce it as a quotient of the braid group, which makes it most natural for the computations we care about, and we refer to Appendix~\ref{app:homflypthecke} and to Mathas' book~\cite{Mathas_book} for further details and context.

\begin{definition}
  Let $\He_n$ be the $\Z[q^{\pm 1}]$ module obtained from the group algebra of the braid group as follows:
  \[
      \He_n=\frac{\Z[q^{\pm 1}][\Br_n]}{(\langle \sigma_i-1)(\sigma_i+q^2)\rangle}
  \]
\end{definition}

We will denote by $\psi:\Br_n\rightarrow \He_n$ the natural induced map. We quickly review some useful basic properties of the Hecke algebra that will be instrumental to us.

\begin{lemma} \label{lem:HeckeBasis}
  Denote by $T_{s_i}$ the image of the i-th generator $\sigma_i\in \Br_n$. Then there is a {\em well-defined} map (of sets):
  \begin{align*}
    \Perm_n&\rightarrow \He_n \ \ \ \ \ & \ \ \ \ \
    s=s_{i_1}\cdots s_{i_k} & \mapsto T_s:=T_{s_{i_1}}\cdots T_{s_{i_k}}.
  \end{align*}
  Furthermore the image $\{T_s\}_{s\in \symG_n}$ is a basis of the Hecke algebra as a $\Z[q^{\pm 1}]$-module.
\end{lemma}
The content of the above lemma is that the image of $s$ does not depend on the word chosen to write it. This is simply because the $T_s$ satisfy the same braiding relations as the elements in the symmetric group. To our eyes, the key property is that we have a basis indexed by permutations, which can be efficiently handled. For computation purposes, we will also take advantage of the fact that the structure constants of this basis are quite simple, as expressed in the following classical lemmas~\cite[Theorem 1.11]{Mathas_book}.

\begin{lemma} \label{lem:prod_hecke}
  \[
  \forall w\in \symG_n, \; \forall i,\;  T_{s_i}T_w=
  \left\{ \begin{array}{ccc} T_{s_iw} &  \text{if}  & \ell(s_iw)>\ell(w); \\
  (1-q^2)T_w +q^2T_{s_iw} & \text{if} & \ell(s_iw)<\ell(w).\\
  \end{array}\right.
  \]  
\end{lemma}

\begin{lemma} \label{lem:inv}
  \[
  \forall w\in \symG_n, \; \forall i,\;  T_{s_i}^{-1}T_w=
  \left\{\begin{array}{ccc}
 (1-q^{-2})T_w +q^{-2}T_{s_iw} & \text{if} & \ell(s_iw)>\ell(w); \\
  T_{s_iw} & \text{if} & \ell(s_iw)<\ell(w).\\
  \end{array}\right.
  \]  
\end{lemma}

When we will want to compute knot invariants, it will be convenient to take braid closures of braids. Assuming that $P$ is a knot invariant, it is clear that $P(\hat{b_1b_2})=P(\hat{b_2b_1})$. This remark motivates the introduction an intermediary object $\Tr(\He_n)$ which is obtained from $\He_n$ by formally adding the relation $b_1b_2=b_2b_1$. This allows to break the computation into a first step $\He_n\rightarrow \Tr(\He_n)$, where elements are reduced as much as possible using the trace relation $b_1b_2=b_2b_1$, before using information about the specific invariant one want to compute to go from $\Tr(\He_n)$ to $\Z[q^{\pm 1}]$.
Note that because $\psi \colon B_n \to \He_n$ is a representation, for any two braids $b_1,_2' \in B_n$, we have $\Tr(\psi(b_1b_2)=\Tr(\psi(b_1)\psi(b_2)) = \Tr(\psi(b_2)\psi(b_1)) = \Tr(\psi(b_2b_1))$.

Finally, when searching for counterexamples, we will want to measure how far an element of $\He_n$ is from the identity. We will use the following notion, that we adapt from~\cite{GWY}.

\begin{definition}
  Let $x=\sum_w a_wT_w\in \He_n$. We define its \emph{projective length}:
  \[
  \projlength(x)=\max\{\deg(a_w),\; w\in \symG_n\}-\min\{\val(a_w),\; w\in \symG_n\},
  \]
  where $\deg$ and $\val$ are the degree and valuation of a Laurent polynomial in the variable $q$.
\end{definition}

Most of the notions we have defined here enjoy similar definitions with $\Z/p\Z$ in place of $\Z$. We will use this context in Section~\ref{sec:non-faithfulness}.

\section{Computing the Hecke representation of the braid group}
\label{sec:comp_hecke}
The goal of the algorithm is to compute the image $\psi(b)$ in the Hecke algebra $\He_n$ of a braid $b=\sigma_{i_c}^{\epsilon_c}\cdots \sigma_{i_1}^{\epsilon_1} \in B_n$, given as an expression of the braid in the generators $\{\sigma_i\}$ of $B_n$.

Our algorithm is iterative, and maintains at step $j$ the image in $\He_n$ of the braid $b_{(j)} \colon= \sigma_{i_j}^{\epsilon_j}\cdots \sigma_{i_1}^{\epsilon_1}$, made of the first $j$ crossings of the input braid.

\subsection{Data structure and initialization}

For a braid in $B_n$, we represent its image in $\He_n$ by an array of coordinates (Laurent polynomials) in the basis $\{T_w\}$. Specifically, a vector $x = \sum_{w \in \Perm_n} a_w T_w \in \He_n$, with $a_w \in \Z[q^{\pm 1}]$, is represented by a length $n!$ array $\hecke$, such that $\hecke[j] = a_w$ whenever the permutation $w$ has index $F_{L(w)} = j$. 

Recall that the one-line word $w(1)\cdots w(n)$ of a permutation $w$, and its index $F_{L(w)}$ can be obtained from one another in $O(n\log n)$ operations. Note also that the maximal value $n!-1$ of an index is below $2^{64}$ when $n \leq 20$, which will always be the case in this article ; permutation indices are stored as standard integers. 

The image of $\sigma_{i_1}$ is $T_{s_{i_1}}\in \He_n$, and Lemma~\ref{lem:inv}, shows that the image of $\sigma_{i_1}^{-1}$ is $q^{-2}T_{s_{i_1}}+(1-q^{-2}) \in \He_n$. 
Consequently, if $\epsilon_1=1$, at the initialization, only the coefficient corresponding to $s_{i_1}$ is $1$, all other ones being $0$. If $\epsilon_1=-1$, then we have the coefficient $1-q^{-2}$ at index $0$, $q^{-2}$ at the index representing $s_{i_1}$, and $0$ elsewhere.

\subsection{Inductive step}
We suppose that at the beginning of step $j$ of the algorithm, $\hecke$ stores the coordinates of the image of $\sigma_{i_{j-1}}^{\epsilon_{j-1}}\cdots \sigma_{i_1}^{\epsilon_1}$, that we update with the next generator $\sigma_{i_{j}}^{\epsilon_{j}}$.

Since $\psi$ is a group homomorphism, we have $\psi(b_{(j)}) = T^{\epsilon_j}_{\sigma_{i_j}} \psi(b_{(j-1)})$. We consequently need to update the coordinates of $\hecke$ following the local relations described in Lemmas~\ref{lem:prod_hecke} and~\ref{lem:inv}. We do so in the single pass through $\hecke$, without allocating more than $O(1)$ Laurent polynomials in memory:

\begin{algorithm}[H]
\SetAlgoLined
    \For{$i \leftarrow 0$ \KwTo $n!-1$}{
      compute $w$ such that $F_{L(w)} = i$\;  
        \If{$\ell(s_{i_j} w) > \ell(w)$}{
          set $\idx_w \leftarrow F_{L(w)}$;$\ $ $\idx_{s_{i_j} w} \leftarrow F_{L(s_{i_j} w)}$;$\ $ $a_w \leftarrow \hecke[\idx_w]$;$\ $ $a_{s_{i_j}w} \leftarrow \hecke[\idx_w]$\;
          \If{$\epsilon_j = +1$}{
            set $\hecke[\idx_w] \leftarrow q^2 a_{s_{i_j} w}$; $\ $ 
            set $\hecke[\idx_{s_{i_j} w}] \leftarrow a_w + (1-q^2) a_{s_{i_j} w}$\;
          }
          \Else{
            set $\hecke[\idx_w] \leftarrow (1-q^{-2}) a_{w} + a_{s_{i_j} w}$; $\ $ set $\hecke[\idx_{s_{i_j} w}] \leftarrow q^{-2} a_w$\;
          }
        }
    }
\caption{From $\hecke$ representing $\psi(b_{(j-1)})$ to representing $\psi(b_{(j)})$}
\label{algo:induction}
\end{algorithm}

\subsection{Correctness}
\label{subsec:correctness}
The correctness of the initialization phase is guaranteed by Lemmas~\ref{lem:prod_hecke} and~\ref{lem:inv}. We focus on the inductive step. 

Consider any permutation $w \in \symG$. According to Lemmas~\ref{lem:prod_hecke} and~\ref{lem:inv}, both $T^{\pm 1}_{s_{i_j}} T_w$ and $T^{\pm 1}_{s_{i_j}} T_{s_{i_j} w}$ give linear combinations of $T_w$ and $T_{s_{i_j} w}$, because $s_{i_j}(s_{i_j} w) = w$. Additionally, for any permutation $w'\notin \{w,s_{i_j} w\}$, the product $T^{\pm 1}_{s_{i_j}} T_{w'}$ contains no term $T_w$ or $T_{s_{i_j} w}$. Thus the pair of coordinates of $\hecke$ corresponding to permutations $w$ and $s_{i_j} w$ can be updated independently from the rest of the array, and their new values is a combination of their previous values. We update them at once in the core of the (outer) {\tt if} loop of Algorithm~\ref{algo:induction}.

Finally, the condition of the (outer) {\tt if} loop guarantees that we proceed to the update once. Indeed, the length $\ell$ being the number of inversions, we have $\ell(s_{i_j} w) \in \{\ell(w)-1,\ell(w)+1\}$. In consequence, considering the two permutations $w$ and $(s_{i_j} w)$, multiplying by $s_{i_j}$ exactly increases the length once and decreases the length once.

In light of the above, the induction algorithm does compute $\psi(b_{(j)})$ from $\psi(b_{(j-1)})$ and, at the end of the algorithm, we are left with $\psi(b)=\sum_{w\in \symG_n} a_w T_w$.

\subsection{Complexity}

\begin{theorem} \label{thm:complexity_hecke}
Given $b \in B_n$ represented by $N$ generators ($N$ crossings), the algorithm above computes the Hecke representation $\psi(b)=\sum_{w\in \symG_n} a_w T_w$ in $O(n! (n\log n) \times N)$ operations and $O(n! \times N)$ algebraic operations in $\Z[q^{\pm 1}]$, storing $n! + O(1)$ Laurent polynomials in memory.
\end{theorem}

\begin{proof}
The initialization of the vector with the first crossing requires $O(n\log n)$ operations to convert permutations into their indices, and $O(1)$ arithmetic operations in $\Z[q^{\pm 1}]$ to initialize a constant number of entries in $\hecke$. The induction step reads the array $\hecke$ (of length $n!$) once, and every time converts a constant number of permutations into their indices ($O(n\log n)$) and performs at most $O(1)$ operations in $\Z[q^{\pm 1}]$. Finally, we run the induction exactly $n$ times.
\end{proof}

\begin{remark}
Note that the algorithm above is rather simple, and the constant hidden by the big-O notation is rather small. This matters for running times in practice ; see Section~\ref{sec:experimental_performance}.
\end{remark}

\begin{remark}\label{rem:parallel}
Note that the main {\bf for} loop of Algorithm~\ref{algo:induction} can be naturally parallelized, considering that pairs of indices in $\hecke$, corresponding to permutations $w$ and $s_{i_j}w$, where $s_{i_j}$ is the transposition corresponding to the new braid generator $\sigma_{i_j}^{\epsilon_j}$, are treated independently from other entries. We use a parallel and a non-parallel implementation of the {\bf for} loop in Section~\ref{sec:experimental_performance}.   
\end{remark}

\section{Computing the HOMFLY-PT polynomial}
\subsection{The HOMFLY-PT polynomial}

The HOMFLY-PT polynomial~\cite{HOMFLY,PT} is a famous 2-variable knot invariant, that can be specialized to both the Jones polynomial and the Alexander polynomial. If one starts from a knot presented as a braid closure $\hat{b}$ (see Figure~\ref{fig:intro_fig}), then, up to renormalization, the invariant is computed by taking a trace of the image of the braid $b$ in the Hecke algebra. We now make this more precise, and we refer to Appendix~\ref{app:homflypthecke} for another presentation of the invariant.

As in Section~\ref{sec:comp_hecke}, we start from a braid $b=\sigma_{i_c}^{\epsilon_c}\cdots \sigma_{i_1}^{\epsilon_1}$, and pre-compute the array $\hecke$ encoding the Hecke representation $\psi(b) = \sum_{w\in \Perm_n} a_wT_w$ with the algorithm of Theorem~\ref{thm:complexity_hecke}.

Our strategy is to simplify, in a single pass through $\hecke$, the entries of the array, using the defining property of the trace, \ie, $\Tr(\psi(bb')) = \Tr(\psi(b'b))$ for $b,b'$ braids in $B_n$. We do not simplify the array all the way down, but instead only keep entries corresponding to permutations $w$ with a simple form (called {\em annularly reduced}), and for which we know how to compute the participation to the final HOMFLY-PT polynomial.

Before describing the algorithm, we introduce some definitions and lemmas.

\begin{definition} \label{def:annred}
An element $w\in \Perm_n$ is said to be \emph{annularly reduced} (see Figure~\ref{fig:annRed}) if it is a product of cycles $r_i$ so that:
\begin{itemize}
\item $r_i$ has support $(j_i,j_i+1,\cdots,j_i+l_i)$ for some $j_i$ and $l_i$;
\item $r_i=(j_i+l_i,j_i,j_i+1,\cdots, j_i+l_i-1)$.
\end{itemize}
\end{definition}

\begin{figure}[t]
\centering
\[
\begin{tikzpicture}[anchorbase,scale=.6]
\draw (0,0) to [out=90,in=-90] (2,2);
\draw (1,0) to [out=90,in=-90] (0,2);
\draw (2,0) to [out=90,in=-90] (1,2);
\draw (3,0) to [out=90,in=-90] (6,2);
\draw (4,0) to [out=90,in=-90] (3,2);
\draw (5,0) to [out=90,in=-90] (4,2);
\draw (6,0) to [out=90,in=-90] (5,2);
\draw (7,0) to [out=90,in=-90] (9,2);
\draw (8,0) to [out=90,in=-90] (7,2);
\draw (9,0) to [out=90,in=-90] (8,2);
\end{tikzpicture}
\]
\caption{An annularly reduced element in $\Perm_{10}$}
\label{fig:annRed}
\end{figure}

Our algorithm reduces braids to annularly reduced ones, using the following lemma.

\begin{lemma} \label{lem:lowerOrder}
  Consider $w\in \Perm_n$, and define, if it exists:
\[
i_0:=\min\{i\;|\; w(i)<w(i)-1\}.
\]
Then $s_{i_0-1}ws_{i_0-1}$ and $ws_{i_0-1}$ are lower than $w$ in the lexicographic order.
\end{lemma}

\begin{proof}
We first argue that $w=w's_{i_0-1}$ as a reduced word. Indeed, from the definition of $i_0$, $w(i_{0}-1)\geq i_{0}-2$ and thus $w(i_{0}-1)\geq w_{i_0}$. The relative order of $i_{0}-1$ and $i_0$ is exchanged under $w$ and thus $w=w's_{i_{0}-1}$ with $\ell(w)=\ell(w')+1$. We are thus in the following situation:
\[
\begin{tikzpicture}[anchorbase,scale=.7]
\foreach \x in {-1,1,2,4}
    \node at  (\x,0) {$\bullet$};
\foreach \x in {-1,1,2,4}
    \node at  (\x,3) {$\bullet$};
\draw (2,0) to [out=90,in=-90]  node [at start, below] {$i_0$} node [at end, above] {$i_0$} (-1,3);
\draw (1,0) to [out=90,in=-90]  node [at start, below] {$i_0-1$} node [at end, above] {$i_0-1$} (4,3);
\node at (1,3.4) {$k$};
\node at (2,3.4) {$l$};
\node at (0,0) {$\cdots$};
\node at (3,0) {$\cdots$};
\node at (0,3) {$\cdots$};
\node at (3,3) {$\cdots$};
\node [rectangle, draw] at (1.5,-1.5) {$w$};
\end{tikzpicture}
\qquad
\begin{tikzpicture}[anchorbase,scale=.7]
\foreach \x in {-1,1,2,4}
    \node at  (\x,0) {$\bullet$};
\foreach \x in {-1,1,2,4}
    \node at  (\x,2) {$\bullet$};
\foreach \x in {-1,1,2,4}
    \node at  (\x,3) {$\bullet$};
\draw (2,0) to [out=90,in=-90]  node [at start, below] {$i_0$}  (4,2) -- (4,3) node [at end, above] {$i_0$};
\draw (1,0) to [out=90,in=-90]  node [at start, below] {$i_0-1$}  (-1,2) -- (-1,3) node [at end, above] {$i_0-1$};
\draw (1,2) to [out=90,in=-90] node [at end,above] {$k$} (2,3) ;
\draw (2,2) to [out=90,in=-90]  node [at end,above] {$l$} (1,3);
\node at (0,0) {$\cdots$};
\node at (3,0) {$\cdots$};
\node at (0,3) {$\cdots$};
\node at (3,3) {$\cdots$};
\node [rectangle, draw] at (1.5,-1.5) {$s_{i_0-1}w'$};
\end{tikzpicture}
\qquad
\begin{tikzpicture}[anchorbase,scale=.7]
\foreach \x in {-1,1,2,4}
    \node at  (\x,0) {$\bullet$};
\foreach \x in {-1,1,2,4}
\node at  (\x,2) {$\bullet$};
\node at (1,3.4) {\vphantom{$k$}};
\draw (2,0) to [out=90,in=-90]  node [at start, below] {$i_0$}   node [at end, above] {$i_0$} (4,2);
\draw (1,0) to [out=90,in=-90]  node [at start, below] {$i_0-1$} node [at end, above] {$i_0-1$} (-1,2) ;
\node at (0,0) {$\cdots$};
\node at (3,0) {$\cdots$};
\node at (0,2) {$\cdots$};
\node at (3,2) {$\cdots$};
\node at (1,2.4) {$k$};
\node at (2,2.4) {$l$};
\node [rectangle, draw] at (1.5,-1.5) {$w'$};
\end{tikzpicture}
\]

In both cases the lexicographic order gets lowered.
\end{proof}

The second piece of argument we will need is the following one.

\begin{lemma} \label{lem:annRed}
If $w\in \Perm_n$ is so that $w(i)\geq i-1$ for all $i$, then $w$ is annularly reduced.
\end{lemma}
\begin{proof}
One first considers only the strands that move to the left. By hypothesis, those can only move one step to the left:
\[
\begin{tikzpicture}[anchorbase,scale=.7]
\foreach \x in {1,2,...,11}
         {\node at (\x,0) {$\bullet$};}
\foreach \x in {1,2,...,11}
         \node at (\x,2) {$\bullet$};
\foreach \x in {2,3,6,7,8,10,11}
         \draw (\x,0) to [out=90,in=-90] (\x-1,2);
\end{tikzpicture}
\]
We are left with relating free ends with strands that can only go straight up or right. There is a unique solution, that consists of pairing the leftmost free end at the bottom with the leftmost free end at the top, and so on, which yields an annularly reduced element.
\end{proof}

We refer to Definition~\ref{def:HomflyPt} in the appendix for a precise definition of the HOMFLY-PT polynomial $P$. For our purpose, it suffices to say that the polynomial $P(\hat{b})$ can be obtained from the vector $\hecke$ by:
\begin{itemize}
  \item reducing the vector $\hecke$ in $\Tr(\He_n)$ using trace moves $b_1b_2=b_2b_1$. This leaves us with $\sum_{w\in \mathcal{A}}a'_wT_w$ where $\mathcal{A}$ denotes the set of annularly reduced words in $\Perm_n$ from Definition~\ref{def:annred};
\item evaluating annularly reduced elements on $\prod_{i}(a^{1-l_i}q^{l_i-1}\frac{a-a^{-1}}{q-q^{-1}})$ (with the notations of Definition~\ref{def:annred}). See Lemma~\ref{lem:s_unknot} for a justification.
\end{itemize}

\subsubsection{Simplifying the trace} \label{sec:homflypt_step2}

We now want to simplify the vector $\hecke$ in a trace process. Given a basis element $T_w$, we will want, when possible, to simplify its image $\Tr(T_w)\in \Tr(\He_n)$. We will do that by ordering elements in $w$ by the lexicographic order $\leq_{\operatorname{lex}}$ of their one-line word $w(1)\cdots w(n)$, so that the simplification can be achieved in one scan through the set $\{T_w\}_{w\in \Perm_n}$.

\begin{remark}
Two elements that are annularly reduced might still be conjugate: the order of the cycles can be permuted. This results in keeping more non-zero entries in the vector we consider than necessary, but has the advantage of yielding a particularly simple algorithm.
\end{remark}

We traverse the array $\hecke$ from right to left, \ie, considering permutations $w$ with {\em decreasing} position in lexicographic order: for every such $w$, 
\begin{enumerate}
\item look for the first $i_0$ so that $w(i_0)<i_0-1$.
\item If there is no such $i_0$ then $w$ is already annularly reduced (Lemma~\ref{lem:annRed}).
\item Otherwise it can be written $w=w's_{i_0-1}$. Replace $T_w$ by $T_{i_0-1}T_w'$ which expresses as a single or a sum of two basis elements (Lemma~\ref{lem:prod_hecke}). In both cases, the elements that appear are strictly lower in the lexicographic order (Lemma~\ref{lem:lowerOrder}), and the coefficient of $T_{i_0-1}T_w'$ is removed and added to entries of $\hecke$ to the left of $w$. 
\end{enumerate}

\subsubsection{Evaluating the polynomial}

We now are left with $v=\sum_{w\in \mathcal{A}}a'_wT_w$. The result we are looking for is:
\[
P(\hat{b})=\sum_{w\in \mathcal{A}}a'_wP(\hat{b_w})\; \text{($b_w$ is the positive lift of $w\in \Br_n$).}
\]

The value for $P(\hat{b_{w}})$ above is recorded in the following lemma.
\begin{lemma} \label{lem:ValueAnnRed}
  The HOMFLY-PT polynomial of an annularly reduced braid with cycles of lengths $l_1,\cdots, l_k$ is $\prod_{i=1}^k(a^{1-l_i}q^{l_i-1}\frac{a-a^{-1}}{q-q^{-1}})$.
\end{lemma}
\begin{proof}
One takes the product of each cycle value, computed in the appendix (Lemma~\ref{lem:s_unknot}).
\end{proof}

The last step of the algorithm thus consists in replacing $T_w$ by the product of the values of the HOMFLY-PT polynomial on the cycles of $T_w$, and taking the sum for all $w$'s, in a single pass through the array $\hecke$.

\subsection{Complexity and correctness}
The correctness of the algorithm is guaranteed by the diverse lemmas proved along the way.

\begin{theorem}\label{thm:complexity_trace}
From the Hecke representation of a braid $b \in B_n$, input as an array $\hecke$ of length $n!$, of Laurent polynomials, in the $\{T_w\}_w$ basis, evaluating the HOMFLY-PT polynomial takes $O(n! (n\log n))$ operations and $O(n! \cdot n)$ algebraic operations, storing $n! + O(1)$ algebraic elements. 
\end{theorem}

\begin{proof}
  Recall from Section~\ref{sec:homflypt_step2} that we reduce the vector to only keep annularly reduced elements in a single pass on the array $\hecke$ of length $n!$. Computing $w$ from its index and finding the first $i_0$ so that $w(i_0)<i_0-1$ takes $O(n\log n)$ operations. Then one finds $w'$ so that $w=w's_{i_0-1}$. This amounts to computing $ws_{i_0-1}$, which costs $O(1)$ operations. Then computing $T_{i_0-1}T_{w'}$ costs two multiplications in $\Z[q^{\pm 1}]$. In total, for this phase, we use $O(n! n \log n)$ combinatorial operations and $O(n!)$ arithmetic operations in $\Z[q^{\pm 1}]$.

  Finally, thanks to Lemma~\ref{lem:ValueAnnRed}, if suffices to know the number of cycles and their lengths to evaluate an annularly reduced element in $O(n)$ arithmetic operations in $\Z[q^{\pm 1}]$ at most (there are at most $n$ disjoint cycles). Identifying cycles and their lengths is done directly by reading the one-line word of a permutation $w$ ($O(n \log n)$ operations to compute the one-line word from w's index, then $O(n)$ to identify cycles), yielding a last step in $O(n! \cdot n \log n)$ combinatorial operations and $O(n! \cdot n)$ arithmetic operations in $\Z[q^{\pm 1}]$. Summing the complexity gives the theorem.
\end{proof}

\begin{theorem}\label{thm:complexity_homfltypt}
Given a braid $b \in B_n$ with $N$ crossings, the above algorithm computes its HOMFLY-PT polynomial in $O(n! (n\log n) \times N)$ operations and $O(n! \times N)$ algebraic operations, storing $n! + O(1)$ algebraic elements. 
\end{theorem}

\begin{proof}
Summing the complexity of Theorem~\ref{thm:complexity_hecke} and Theorem~\ref{thm:complexity_trace}, all steps of the algorithm are covered; assuming $n \leq N$, Theorem~\ref{thm:complexity_hecke} dominates the computation.
\end{proof}

\begin{figure}[t]
\centering
\includegraphics[width=0.8\textwidth]{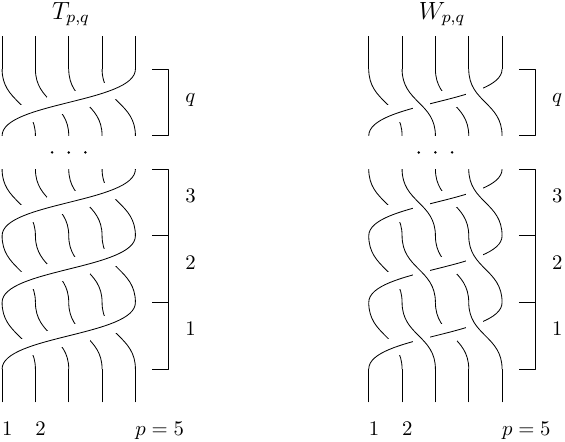}
\caption{Torus braid $T_{p,q}$ and weaving braid $W_{p,q}$.}
\label{fig:TWbraids}
\end{figure}

\subsection{Experimental performance}
\label{sec:experimental_performance}
We have compared our code with the advanced {\tt C++} library {\tt Regina}~\cite{regina}, implementing Burton's state of the art algorithm~\cite{burton:LIPIcs.SoCG.2018.18} for the HOMFLY-PT polynomial. We have run our experiments on an 8 cores, 32Gb RAM computer; parallel implementations use the library {\tt Intel::TBB}. We restrict our attention to braids whose closures lead to knots, because {\tt Regina} implements the HOMFLY-PT polynomial on knots exclusively. The computation of the Hecke representation, which largely dominates the timings for our computation of the HOMFLY-PT polynomial, can be implemented in parallel, as described in Remark~\ref{rem:parallel}, while {\tt Regina} does not admit a parallel implementation and the algorithm it implements may not parallelize as naturally as ours.

We have generated $p$-strands torus braids $T_{p,q}$ and $p$-strands weaving braids $W_{p,q}$, in their standard form (see Figure~\ref{fig:TWbraids}), keeping only those parameters $q$ for which the closure is a knot. We have used them as is in our algorithm. {\tt Regina} working with knot diagrams and treewidth, we have run simplification heuristics and the construction of the tree decomposition, before computing the HOMFLY-PT polynomial (and before starting the timer). However, for these knots, the braid presentation is expected to be minimal for both size and parameter. 

\begin{figure}[t]
\centering
\includegraphics[width=1\textwidth]{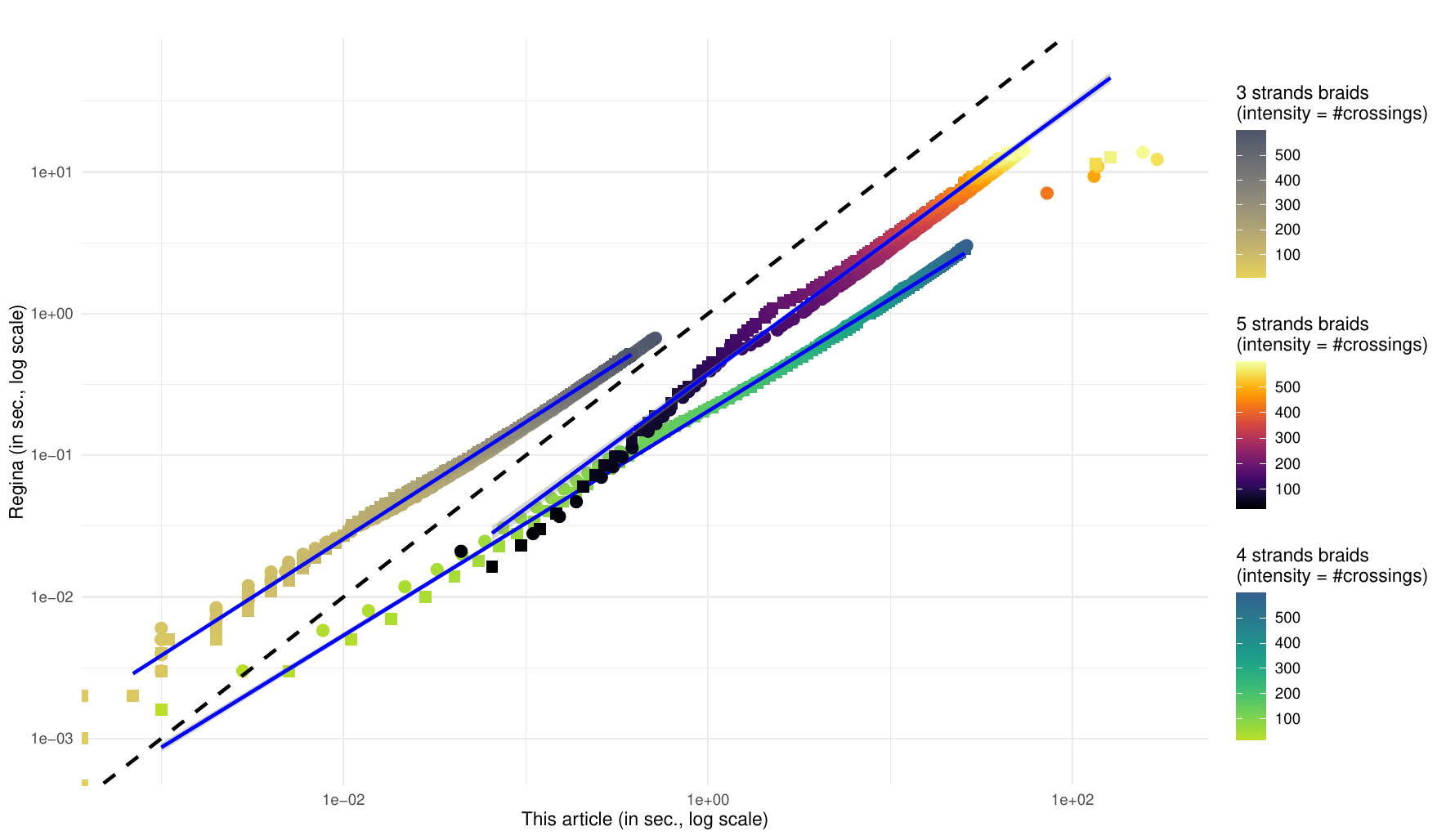}
\caption{Timings of our algorithm (without parallelism) and {\tt Regina}. Circle points represent torus braids, and square points represent weaving braids. The dashed line is $x=y$.}
\label{fig:timings}
\end{figure}

In Figure~\ref{fig:timings} we compare the running time of our algorithm {\em not using parallelism}, against the algorithm from {\tt Regina} for torus braids $T_{p,q}$ and weaving braids $W_{p,q}$, for $3 \leq p < 6$ and all values of $q = p+1, \ldots, q_{\max}$, such that the closures of the braids lead to knots (\ie, $\gcd(p,q) = 1$), with $q_{\max}$ picked such that the braids $T_{p,q_{\max}}$ and $W_{p,q_{\max}}$ have approximately $600$ crossings. The results are very similar for torus and weaving braids, highlighting the fact that the running times depends mostly on the combinatorics of the input diagrams, and not so much on the algebraic computations. Our implementation consistently outperforms {\tt Regina} on braids with more than 4 strands, by a factor ranging from 3 (shorter braids) to 9 (longer braids) for braids on 4 strands, a factor ranging from 3 to 4 for 5-braids. 

In Figure~\ref{fig:timings}, {\tt Regina} is faster on the smaller 3-strands examples, by a factor 3 on shorter braids, and a smaller factor 1.3 on the larger braids ; all these timings are below 0.3 seconds, and may be explained by a slower initialization of data on our side. The global tendency is that our algorithm, even on its non-parallel form, gets better on longer braids (the slopes $\alpha$ of interpolating lines is $0.78 \leq \alpha \leq 0.8$ on all examples) and on braids with more strands. 

On braids with $\geq$ 6 strands, {\tt Regina} starts to {\tt swap} due to excessive memory consumption, leading to impractical running times ; to the contrary, our implementation continues to be practical, Table~\ref{table:timings} and see Figure~\ref{fig:parallel} for much larger examples. 

In Figure~\ref{fig:parallel}, we compare the performance of our implementation with and without parallelism. Recall that the parallelization is on the array $\hecke$ of length $n!$ (working with an $n$-strands braid). In consequence, we see the gain of parallelism increasing with the number of strands: first, the gain range from a factor $1.6$ to $2.4$ on 3-strands, from $3.3$ to $4.54$ 4-strands, and from $3.8$ to $6.25$ on 5-strands. On braids with $6,7,8$-strands, the gain ranges between a factor $5$ to $7.14$, near the optimal 8 (on 8 cores); see Figure~\ref{fig:parallel}.

Finally, the memory consumption of our implementation is much more frugal than {\tt Regina}, and remains very applicable for large number of strands. The memory cost depends mostly on the number of strands; however, for a fixed number of strands, longer braids tend to consume slightly more memory because the Laurent polynomials stored in the array $\hecke$ tend to become more and more complicated ; we present the memory consumption of longer braids in Table~\ref{table:timings}).

\begin{table}
\centering
\begin{tabular}{rrrrr}
Braid & Num. of crossings & Timing parallel (8 cores) & Timing non-parallel & Peak memory \\
 $T_{6,41}$ & 205 & 1.6 sec. & 9.6 sec.& 8.1 Mb \\
 $T_{6,119}$ & 595 & 13.4 sec. & 95.1 sec. & 23.5 Mb \\
 $T_{7,36}$  & 216 & 10.8 sec. & 74 sec. &  45.6 Mb \\
 $T_{7,99}$  & 594 & 88.9 sec.  & 643 sec. &  109 Mb \\
 $T_{8,29}$ & 203 & 74.5 sec. & 502 sec. & 257 Mb \\
 $T_{8,49}$ & 343 & 272 sec. & 2265 sec. & 435 Mb \\
\end{tabular}
\caption{Timing and peak memory usage of our algorithm using parallelism on large torus braids (similar results are found for weaving braids).}
\label{table:timings}
\end{table}

In conclusion, our implementation is very competitive for knots represented by braids, even with a fairly large number of strands. It also combines very well with parallelism. We apply our parallel implementation in experimental mathematics in Section~\ref{sec:non-faithfulness} and discover new algebraic properties of the Hecke algebra.

\begin{figure}[h!]
\centering
\includegraphics[width=1\textwidth]{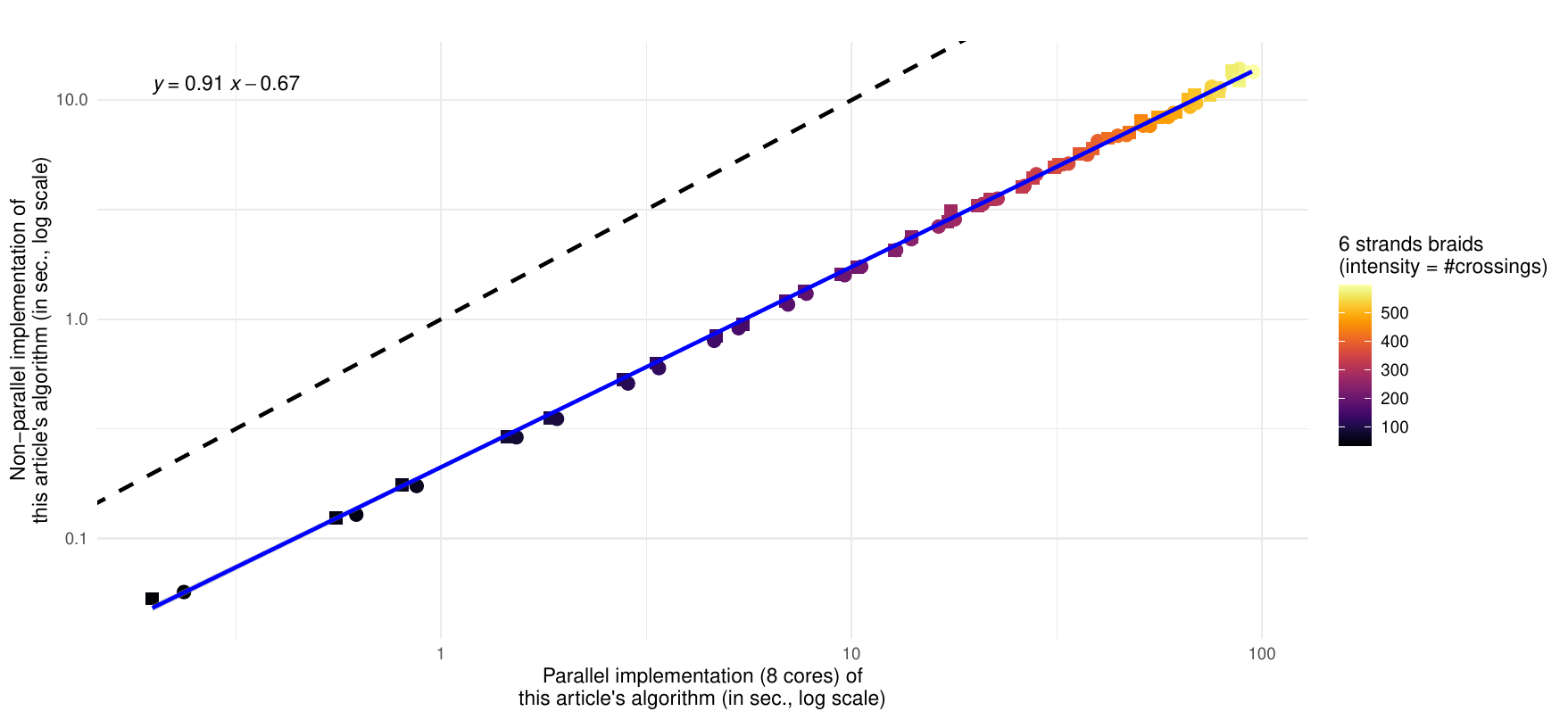}
\bigskip
\includegraphics[width=1\textwidth]{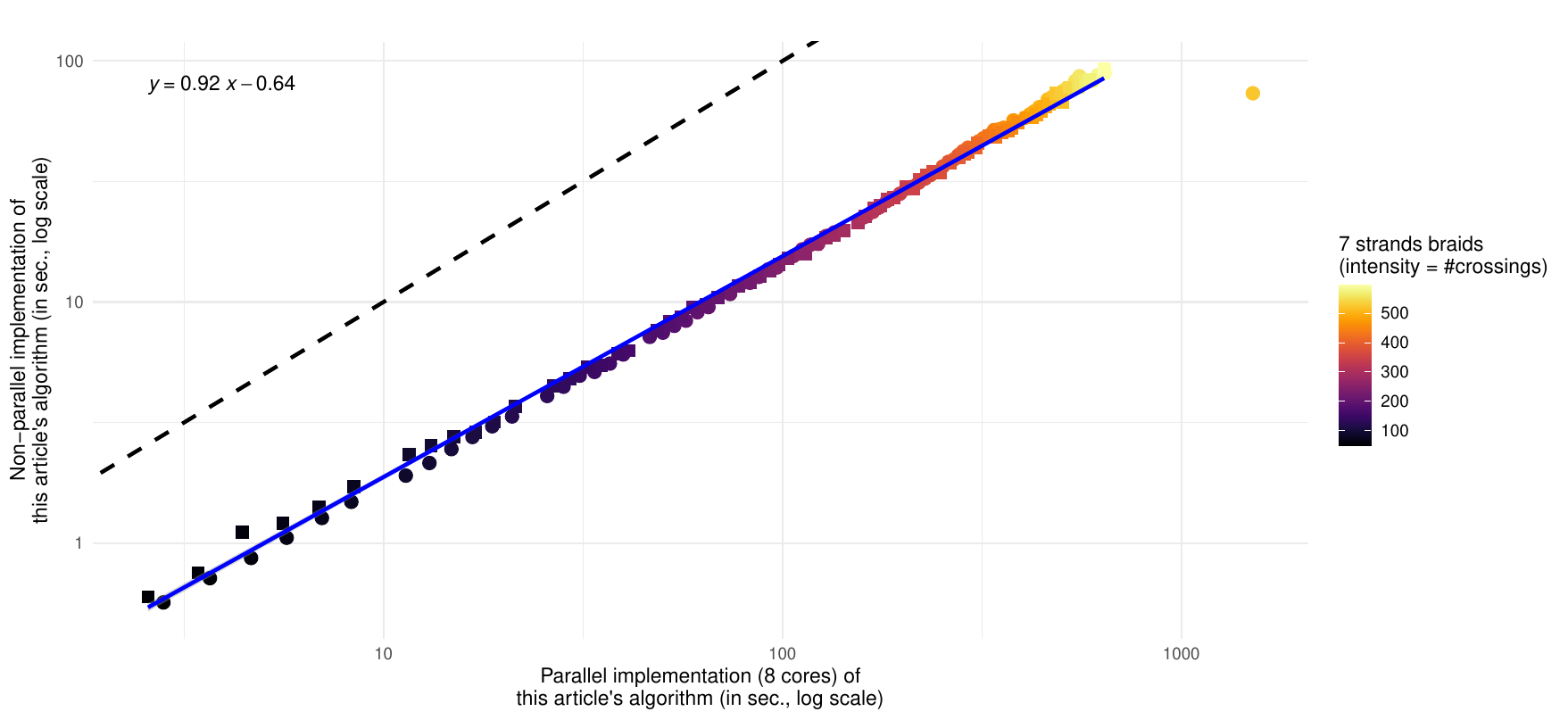}
\bigskip
\includegraphics[width=1\textwidth]{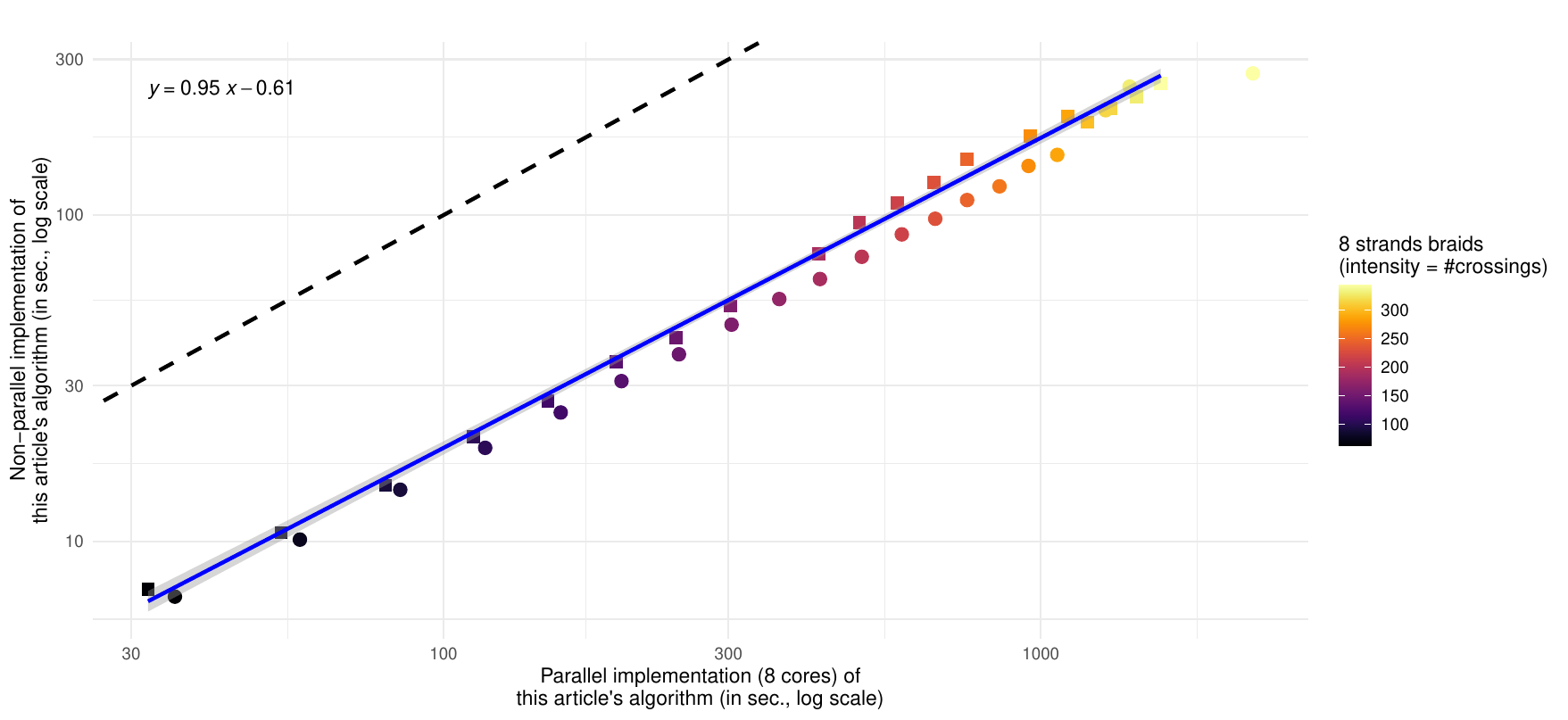}
\caption{Timings of our algorithm with parallelism against our algorithm without parallelism. The interpolation line are near $x=y/7.14$ (we use 8 cores).}
\label{fig:parallel}
\end{figure}

\section{Proving the non-faithfulness of the Hecke representation in $\Z/m\Z$}
\label{sec:non-faithfulness}

With a bucket search procedure, coupled with the fast algorithm described in Section~\ref{sec:comp_hecke}, we have been looking for non-trivial braids with trivial Hecke representation in several contexts.

Following~\cite{Bigelow_Burau,GWY,BQ}, we generate a sample of non-trivial braids from $B_n$ in Garside normal form, and store them in {\em buckets} of bounded size, one for each represented value of {\em augmented projlength} (see below). At iteration $k$ of the search, the buckets store braids of Garside length $k$. For initialization, the Garside length $1$ braids are in bijection with permutations in $\Perm_n \setminus \{ \id, \Delta\}$, and we store them in the buckets. Inductively, if we have generated a sample of Garside length $k$ braids, we attempt to extend each of them into Garside length $k+1$ braids by appending an extra compatible Garside letter. Note that, when attempting to insert a new generated braid, we use a {\em reservoir sampling}~\cite{10.1145/3147.3165} approach to guarantee that we extract an iid sample. In Section~\ref{sec:app_gar}, we describe details on the optimized extension of Garside normal forms we have implemented. In particular, we prove,

\begin{lemma}
Given a braid $b=b_k \ldots b_1 \in B_n$ in Garside normal form, there is an $O(n \log n)$ algorithm to pick a Garside letter $b_{k+1}$ \emph{iid} among all Garside letters compatible with $b_k$.
\end{lemma}

We define the {\em augmented projlength} of a positive braid $b$ to be
\[
  \widehat{\operatorname{projlength}}(b) \colon \sup_{k \leq 0} \{\operatorname{projlength}(\psi(\Delta^k b))\},
\]
and we try to minimize this quantity over the search. To avoid computing the $\operatorname{projlength}$ for infinitely many negative powers of $\Delta$, we (heuristically) stop the computation of $\widehat{\operatorname{projlength}}(b)$ when we encounter a power $k_0 < 0$ for which $\psi(\Delta^k b)$ has a coordinate with negative degree.

The bucket search algorithm, if successful, finds positive braids $b$ (not of the form $\Delta^l$) such that there exists of negative power $k$ of $\Delta$ satisfying: $\Delta^k \cdot b$ has $\projlength = 0$. In all cases where we found such braids, they had an image in the Hecke algebra (in the basis used from Lemma~\ref{lem:HeckeBasis}) with a single non-trivial coordinate $\idx$ equal to $1$. This coordinate corresponding to the positive permutation $s^{-1}$, the braid $s^{-1} \cdot \Delta^k \cdot b$ has trivial Hecke representation.

We have run our searches in parallel on a cluster of 2 CPUs, 128 cores/CPU, and 1024GB RAM.

\subsection{Non-faithful Hecke representations of $B_4$}

The decomposition of $\He_4$ into irreducible $B_4$ representations only contains the Burau representation as possibly faithful summand. The latter being not faithful at $p=2,3,5$ implies the same for the Hecke representation. We indeed found such counterexamples for $p=2$ and $3$, as well as with $p=4$, that we present in Appendix~\ref{app:resultsB4}.

\subsection{Non-faithful Hecke representations of $B_5$}

The same search performed on $B_5$ with base ring $\Z/2\Z$ also yielded a counterexample to faithfulness. To the best of our knowledge, this is a genuinely new and intriguing result.

\begin{theorem} \label{thm:kernelB5}
The braid $s_9^{-1} \Delta^{-9} b$ is in the kernel of the Hecke representation of $\Br_5$ with coefficients in $\Z/2\Z[q^{\pm 1}]$, with :

\begin{align*}
  b = &\sigma_{3} \sigma_{4} \sigma_{2} \sigma_{3} \sigma_{1} \sigma_{2} \sigma_{2} \sigma_{4} \sigma_{2} \sigma_{3} \sigma_{1} \sigma_{2} \sigma_{4} \sigma_{2} \sigma_{3} \sigma_{3} \sigma_{2} \sigma_{3} \sigma_{1} \sigma_{3}
 \sigma_{1} \sigma_{2} \sigma_{2} \sigma_{3} \sigma_{1} \sigma_{3} \sigma_{1} \sigma_{2} \sigma_{2} \sigma_{3} \sigma_{1} \sigma_{3} \sigma_{1} \sigma_{2} \sigma_{4} \sigma_{3}
\\
&   \sigma_{2}\sigma_{3} \sigma_{4} \sigma_{1} \sigma_{4} \sigma_{3} \sigma_{2} \sigma_{1} \sigma_{2} \sigma_{3} \sigma_{4} \sigma_{4} \sigma_{3} \sigma_{2} \sigma_{3} \sigma_{1} \sigma_{2} \sigma_{3} \sigma_{3} \sigma_{2} \sigma_{3} \sigma_{4} \sigma_{1} \sigma_{2} \sigma_{3} \sigma_{4} \sigma_{4} \sigma_{3} \sigma_{4} \sigma_{2} \sigma_{1} \sigma_{2} \sigma_{3} \sigma_{3} \sigma_{4} \sigma_{2}
\\
&\sigma_{1} \sigma_{2} \sigma_{3} \sigma_{4} \sigma_{4} \sigma_{3} \sigma_{4} \sigma_{1} \sigma_{2} \sigma_{3} \sigma_{3} \sigma_{4} \sigma_{2} \sigma_{3} \sigma_{4} \sigma_{1} \sigma_{4} \sigma_{3} \sigma_{4} \sigma_{1} \sigma_{2}
\\
  s_9^{-1} =& \sigma^{-1}_{2} \sigma^{-1}_{3} \sigma^{-1}_{4},\qquad
  \Delta^{-1} = \sigma^{-1}_{4} \sigma^{-1}_{3} \sigma^{-1}_{2} \sigma^{-1}_{1} \sigma^{-1}_{4} \sigma^{-1}_{3} \sigma^{-1}_{2} \sigma^{-1}_{4} \sigma^{-1}_{3} \sigma^{-1}_{4}.
\end{align*}
$b$ is of Garside length $19$, and the resulting counterexample is a braid with $186$ crossings.
\end{theorem}

\subsection{Garside automaton} \label{sec:app_gar}

To generate braids, we use the following finite state automaton.

\begin{definition}
Let $\mathrm{\bf Garside}$ be the finite state automaton with:
\begin{itemize}
\item states, permutations $s$ in $\Perm_n$ (or equivalently elements in $[1,\Delta]$);
\item $s'$-labeled arrows $s\rightarrow s'$ if $\forall s_j\leq_R s'$ then $s_j\leq_L s$.
\end{itemize}
\end{definition}

We create increasingly complicated braids by applying to a Garside normal form $b_k\cdots b_1$ any (possibly randomly chosen) letter from $\mathrm{\bf Garside}(b_k)$. The construction of the automaton can be made very efficient thanks to the following observation.

\begin{lemma} \label{lem:NF_cond}
  Given $s=(s(1),\cdots, s(n))$ and $s'=(s'(1),\cdots,s'(n))$, $s'$ can be applied at the state $s$ if and only if the following property holds:
\[
s'\neq \mathrm{id}\;\text{and}\; \forall i\in \llbracket 1,n \rrbracket \;\text{with}\; s'(i)>s'(i+1),\; s^{-1}(i)>s^{-1}(i+1).
\]
\end{lemma}
\begin{proof}
  This restates the fact that right descents from $s'$ need to be left descents of $s$.
\end{proof}

From a state $s$, we want to generate admissible arrows $s'$. All pairs $\{(i_r,i_r+1)\;|\;\exists l<k,\; i_r=s(k),\; i_r+1=s(l)\}$ partition $\llbracket 1,n\rrbracket$ into intervals $\llbracket i_r+1, i_{r+1}\rrbracket$. Lemma~\ref{lem:NF_cond} prevents two entries from the same interval to cross in $s'$. This suggests the following (see Figure~\ref{fig:GarsideSampling}):
\begin{itemize}
\item for length 1 intervals ($i_r+1=i_{r+1}$), freely choose $s'(i_r)$ in $\llbracket 1, n\rrbracket$ (forming $S_{free}$);
\item images of the other intervals will be disjoint ordered subsets of $\llbracket 1,n\rrbracket\setminus S_{free}$. These choices can be made one interval at a time in any order. Starting from the first interval $\llbracket i_r+1,i_{r+1}\rrbracket$ of length $p'>1$, one chooses a random ordered subset $S_{p'}$ of $p'$ elements in $\llbracket 1, n-p\rrbracket$, and then decide that $s'(i_r+k)$ is the $S_{p'}(k)$-th element in $\llbracket 1,n\rrbracket\setminus S_{free}$. Then one replaces $S_{free}$ by $S_{free}\cup S_{p'}$ and goes to the next interval of size greater than $1$.
\end{itemize}

\begin{figure}[!h]
\[
\begin{tikzpicture}[anchorbase,scale=.5]
\node (1) at (1,0) {};
\node (1p) at (1,3) {};
\node (2) at (2,0) {};
\node (2p) at (2,3) {};
\node (3) at (3,0) {};
\node (3p) at (3,3) {};
\node (4) at (4,0) {};
\node (4p) at (4,3) {};
\node (5) at (5,0) {};
\node (5p) at (5,3) {};
\node (6) at (6,0) {};
\node (6p) at (6,3) {};
\node (7) at (7,0) {};
\node (7p) at (7,3) {};
\node (8) at (8,0) {};
\node (8p) at (8,3) {};
\draw (1) -- (3p.center);
\draw (2) -- (1p.center);
\draw (3) -- (4p.center);
\draw (4) -- (2p.center);
\draw (5) -- (7p.center);
\draw (6) -- (6p.center);
\draw (7) -- (5p.center);
\draw (8) -- (8p.center);
\node at (1,-.2) {\tiny 1};
\node at (2,-.2) {\tiny 2};
\node at (3,-.2) {\tiny 3};
\node at (4,-.2) {\tiny 4};
\node at (5,-.2) {\tiny 5};
\node at (6,-.2) {\tiny 6};
\node at (7,-.2) {\tiny 7};
\node at (8,-.2) {\tiny 8};
%% \node at (2,3.2) {\tiny $i_1$};
%% \node at (3,3.2) {\tiny $i_1+1$};
%% \node at (5,3.2) {\tiny $i_2$};
%% \node [rotate=60] at (6.3,3.5) {\tiny $i_2+1=i_3$};
%% \node at (7,3.2) {\tiny $i_3+1$};
\draw (1.5,3) ellipse (.9 and .1); 
\draw (4,3) ellipse (1.4 and .1); 
\draw (6,3) ellipse (.2 and .1); 
\draw (7.5,3) ellipse (.7 and .1);
\node (1pp) at (1,6) {}; %% for alignment
\draw [decoration={brace,raise=8pt},decorate] (1,0) -- node[left=11pt] {$s$} (1,3);
\node (legend) at (4.5,-2) {\small Start from $s$.};
\node (legend) at (4.5,8) {\vphantom{\small Start from $s$.}};
\end{tikzpicture}
\;\longrightarrow\;
\begin{tikzpicture}[anchorbase,scale=.5]
\node (1) at (1,0) {};
\node (1p) at (1,3) {};
\node (2) at (2,0) {};
\node (2p) at (2,3) {};
\node (3) at (3,0) {};
\node (3p) at (3,3) {};
\node (4) at (4,0) {};
\node (4p) at (4,3) {};
\node (5) at (5,0) {};
\node (5p) at (5,3) {};
\node (6) at (6,0) {};
\node (6p) at (6,3) {};
\node (7) at (7,0) {};
\node (7p) at (7,3) {};
\node (8) at (8,0) {};
\node (8p) at (8,3) {};
\draw (1) -- (3p.center);
\draw (2) -- (1p.center);
\draw (3) -- (4p.center);
\draw (4) -- (2p.center);
\draw (5) -- (7p.center);
\draw (6) -- (6p.center);
\draw (7) -- (5p.center);
\draw (8) -- (8p.center);
\node at (1,-.2) {\tiny 1};
\node at (2,-.2) {\tiny 2};
\node at (3,-.2) {\tiny 3};
\node at (4,-.2) {\tiny 4};
\node at (5,-.2) {\tiny 5};
\node at (6,-.2) {\tiny 6};
\node at (7,-.2) {\tiny 7};
\node at (8,-.2) {\tiny 8};
%% \node at (2,3.2) {\tiny $i_1$};
%% \node at (3,3.2) {\tiny $i_1+1$};
%% \node at (5,3.2) {\tiny $i_2$};
%% \node [rotate=60] at (6.3,3.5) {\tiny $i_2+1=i_3$};
%% \node at (7,3.2) {\tiny $i_3+1$};
\draw (1.5,3) ellipse (.9 and .1); 
\draw (4,3) ellipse (1.4 and .1); 
\draw (6,3) ellipse (.2 and .1); 
\draw (7.5,3) ellipse (.7 and .1);
\node (1pp) at (1,6) {};
\node (2pp) at (2,6) {};
\node (3pp) at (3,6) {};
\node (4pp) at (4,6) {};
\node (5pp) at (5,6) {};
\node (6pp) at (6,6) {};
\node (7pp) at (7,6) {};
\node (8pp) at (8,6) {};
\draw (6p.center) -- (3pp);
\node (legend) at (4.5,-2) {\small Choose $S_{free}$.};
\node (legend) at (4.5,8) {\vphantom{\small Choose $S_{free}$.}};
\end{tikzpicture}
\;\longrightarrow\;
\begin{tikzpicture}[anchorbase,scale=.5]
\node (1) at (1,0) {};
\node (1p) at (1,3) {};
\node (2) at (2,0) {};
\node (2p) at (2,3) {};
\node (3) at (3,0) {};
\node (3p) at (3,3) {};
\node (4) at (4,0) {};
\node (4p) at (4,3) {};
\node (5) at (5,0) {};
\node (5p) at (5,3) {};
\node (6) at (6,0) {};
\node (6p) at (6,3) {};
\node (7) at (7,0) {};
\node (7p) at (7,3) {};
\node (8) at (8,0) {};
\node (8p) at (8,3) {};
\draw (1) -- (3p.center);
\draw (2) -- (1p.center);
\draw (3) -- (4p.center);
\draw (4) -- (2p.center);
\draw (5) -- (7p.center);
\draw (6) -- (6p.center);
\draw (7) -- (5p.center);
\draw (8) -- (8p.center);
\node at (1,-.2) {\tiny 1};
\node at (2,-.2) {\tiny 2};
\node at (3,-.2) {\tiny 3};
\node at (4,-.2) {\tiny 4};
\node at (5,-.2) {\tiny 5};
\node at (6,-.2) {\tiny 6};
\node at (7,-.2) {\tiny 7};
\node at (8,-.2) {\tiny 8};
%% \node at (2,3.2) {\tiny $i_1$};
%% \node at (3,3.2) {\tiny $i_1+1$};
%% \node at (5,3.2) {\tiny $i_2$};
%% \node [rotate=60] at (6.3,3.5) {\tiny $i_2+1=i_3$};
%% \node at (7,3.2) {\tiny $i_3+1$};
\draw (1.5,3) ellipse (.9 and .1); 
\draw (4,3) ellipse (1.4 and .1); 
\draw (6,3) ellipse (.2 and .1); 
\draw (7.5,3) ellipse (.7 and .1);
\node (1pp) at (1,6) {};
\node (2pp) at (2,6) {};
\node (3pp) at (3,6) {};
\node (4pp) at (4,6) {};
\node (5pp) at (5,6) {};
\node (6pp) at (6,6) {};
\node (7pp) at (7,6) {};
\node (8pp) at (8,6) {};
\draw (6p.center) -- (3pp);
\draw (1p.center) -- (2pp);
\draw (2p.center) -- (5pp);
\draw (3p.center) -- (1pp);
\draw (4p.center) -- (4pp);
\draw (5p.center) -- (8pp);
\draw (7p.center) -- (6pp);
\draw (8p.center) -- (7pp);
\draw [decoration={brace,mirror,raise=8pt},decorate] (8,3) -- node[right=11pt] {$s'$} (8,6);
\node (legend) at (4.5,-2) {\small \begin{minipage}{4cm}Choose the other images so that no strands coming out of the same ellipse cross.\end{minipage}};
\node  at (4.5,8) {\vphantom{\small \begin{minipage}{4cm}Choose the other images so that no strands coming out of the same ellipse cross.\end{minipage}}};
\end{tikzpicture}
\]
\caption{Generating descents}
\label{fig:GarsideSampling}
\end{figure}
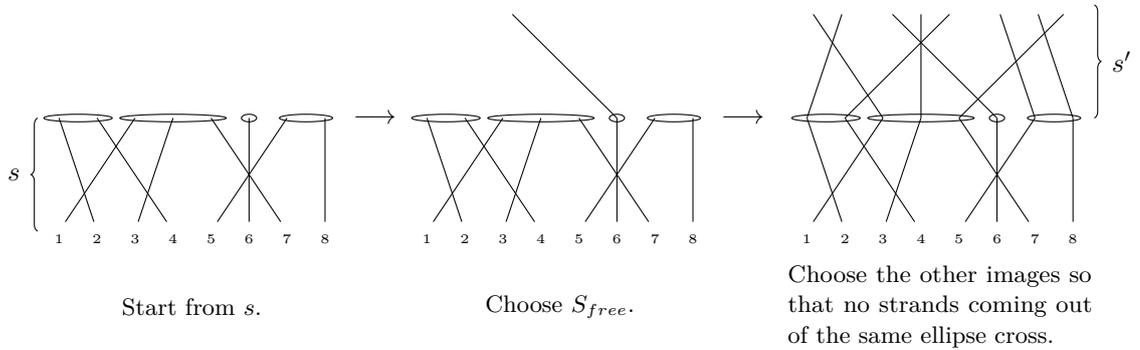

A major advantage of the above construction is the following.

\begin{lemma}
Given a braid $b \in B_n$ in Garside normal form $b_k \ldots b_1$, there is an $O(n \log n)$ algorithm to pick a Garside letter $b_{k+1}$ iid among all Garside letters compatible with $b_k$.
\end{lemma}

\begin{proof}
If one can make uniform random choices of $p$ unordered elements in $\llbracket 1,n\rrbracket$, and uniform random choices of $p$ ordered elements in $\llbracket 1, p'\rrbracket$, then the procedure described above the lemma produces uniform choices in $\mathrm{\bf Garside}(s)$. Picking iid random subsets can be done in time $O(n \log n)$ using, for example, a reservoir sampling algorithm.
\end{proof}

\bibliography{biblio.bib}

@article{BH,
  title={Unknotting number is not additive under connected sum},
  author={Brittenham, M. and Hermiller, S.},
  year={2025},
  note={\href{https://arxiv.org/abs/2506.24088}{arXiv:2506.24088}}
}

@article{Garoufalidis2005,
  title = {Experimental evidence for the Volume Conjecture for the simplest hyperbolic non-2–bridge knot},
  volume = {5},
  ISSN = {1472-2747},
  url = {http://dx.doi.org/10.2140/agt.2005.5.379},
  DOI = {10.2140/agt.2005.5.379},
  number = {1},
  journal = {Algebraic \& Geometric Topology},
  publisher = {Mathematical Sciences Publishers},
  author = {Garoufalidis,  Stavros and Lan,  Yueheng},
  year = {2005},
  month = may,
  pages = {379–403}
}

@article{AIF_0__0_0_A159_0,
     author = {Detcherry, Renaud},
     title = {A quantum obstruction for purely cosmetic surgeries},
     journal = {Annales de l'Institut Fourier},
     year = {2025},
     publisher = {Association des Annales de l{\textquoteright}institut Fourier},
     doi = {10.5802/aif.3673},
     language = {en},
     note = {Online first},
}

@article{Squier,
  title={Matrix representations of Artin groups},
  author={Squier, C. C},
  journal={Proceedings of the American Mathematical Society},
  volume={103},
  number={1},
  pages={49--53},
  year={1988}
}

@article{LX,
  title={On the injectivity of the braid group in the Hecke algebra},
  author={Lehrer, G. I and Xi, N.},
  journal={Bulletin of the Australian Mathematical Society},
  volume={64},
  number={3},
  pages={487--493},
  year={2001},
  publisher={Cambridge University Press}
}

@article{bigelow_IH,
  title={Braid groups and Iwahori-Hecke algebras},
  author={Bigelow, S.},
  journal={Problems on mapping class groups and related topics},
  volume={74},
  pages={285--299},
  year={2006},
  note={\href{https://arxiv.org/abs/1411.5418}{arXiv:1411.5418}}  
}

@article{BQ,
  title={Some remarks about the faithfulness of the {B}urau representation of {A}rtin--{T}its groups},
  author={Bapat, A. and Queffelec, H.},
  year={2024},
  note={\href{https://arxiv.org/abs/2409.00144}{arXiv:2409.00144}}
}

@article{Bigelow_Burau,
  title={The {B}urau representation is not faithful for n= 5},
  author={Bigelow, S.},
  journal={Geometry \& Topology},
  volume={3},
  number={1},
  pages={397--404},
  year={1999},
  publisher={Mathematical Sciences Publishers}
}

@article{Ito,
  title={A kernel of a braid group representation yields a knot with trivial knot polynomials},
  author={Ito, T.},
  journal={Mathematische Zeitschrift},
  volume={280},
  number={1},
  pages={347--353},
  year={2015},
  publisher={Springer},
  note={\href{https://arxiv.org/abs/1402.2028}{arXiv:1402.2028}}
}

@Article{Jones,
  Title                    = {A polynomial invariant for knots via von {N}eumann algebras},
  Author                   = {Jones, V. F. R.},
  Journal                  = {Bull. Amer. Math. Soc. (N.S.)},
  Year                     = {1985},
  Number                   = {1},
  Pages                    = {103--111},
  Volume                   = {12},
  Coden                    = {BAMOAD},
  Doi                      = {10.1090/S0273-0979-1985-15304-2},
  Fjournal                 = {American Mathematical Society. Bulletin. New Series},
  ISSN                     = {0273-0979},
  Mrclass                  = {57M25 (46L10)},
  Mrnumber                 = {766964 (86e:57006)},
  Mrreviewer               = {J. S. Birman},
  Url                      = {http://dx.doi.org/10.1090/S0273-0979-1985-15304-2}
}

@article{KL,
  title={Representations of Coxeter groups and Hecke algebras},
  author={Kazhdan, D. and Lusztig, G.},
  journal={Inventiones mathematicae},
  volume={53},
  number={2},
  pages={165--184},
  year={1979}
}

@article{Jones_Hecke,
  title={Hecke algebra representations of braid groups and link polynomials},
  author={Jones, V. FR},
  journal={Annals of Math.},
  pages={335--388},
  volume={126},
  year={1987}
}

@misc{Morton_code,
  title        = {br9z},
  author       = {H. Morton},
  howpublished = {\url{https://www.liverpool.ac.uk/~su14/knotprogs.html}},
  year         = {1985},
}

@book{Dehornoy_Garside_book,
  title={Foundations of Garside theory},
  author={Dehornoy, P. and Digne, F. and Godelle, E. and Krammer, D. and Michel, J.},
  volume={22},
  year={2015},
  publisher={European Mathematical Society Z{\"u}rich}
}

@article{Garside,
  title={The braid group and other groups},
  author={Garside, F. A},
  journal={The Quarterly Journal of Mathematics},
  volume={20},
  number={1},
  pages={235--254},
  year={1969},
  publisher={Oxford University Press}
}

@article{GWY,
  author       = {J. Gibson AND G. Williamson AND O. Yacobi},
  title        = {{4-Strand {B}urau is Unfaithful Modulo 5}},
  year         = 2023,
  eprint       = {2310.02403v1},
  primaryclass = {math.RT},
  archiveprefix= {arXiv},
    Note={\href{https://arxiv.org/abs/2310.02403}{arXiv:2310.02403}},
}

@article{HOMFLY,
 author = {Freyd, P. and Yetter, D. and Hoste, J. and Lickorish, W. B. R. and Millett, K. and Ocneanu, A.},
 title = {A new polynomial invariant of knots and links},
 fjournal = {Bulletin of the American Mathematical Society. New Series},
 journal = {Bull. Am. Math. Soc., New Ser.},
 issn = {0273-0979},
 volume = {12},
 pages = {239--246},
 year = {1985},
 doi = {10.1090/S0273-0979-1985-15361-3}
}

@article{PT,
 author = {Przytycki, J. H. and Traczyk, P.},
 title = {Invariants of links of {Conway} type},
 fjournal = {Kobe Journal of Mathematics},
 journal = {Kobe J. Math.},
 issn = {0289-9051},
 volume = {4},
 number = {2},
 pages = {115--139},
 year = {1987},
}

@book{Mathas_book,
 author = {Mathas, A.},
 title = {Iwahori-{Hecke} algebras and {Schur} algebras of the symmetric group},
 fseries = {University Lecture Series},
 series = {Univ. Lect. Ser.},
 issn = {1047-3998},
 volume = {15},
 isbn = {0-8218-1926-7},
 year = {1999},
 publisher = {Providence, RI: American Mathematical Society},
}

@Article{RT,
  Title                    = {Ribbon graphs and their invariants derived from quantum groups},
  Author                   = {Reshetikhin, N. Yu. and Turaev, V. G.},
  Journal                  = {Comm. Math. Phys.},
  Year                     = {1990},
  Number                   = {1},
  Pages                    = {1--26},
  Volume                   = {127},
  Fjournal                 = {Communications in Mathematical Physics},
  Url                      = {http://projecteuclid.org/euclid.cmp/1104180037}
}

@article {Alexander,
    AUTHOR = {Alexander, J. W.},
     TITLE = {Topological invariants of knots and links},
   JOURNAL = {Trans. Amer. Math. Soc.},
  FJOURNAL = {Transactions of the American Mathematical Society},
    VOLUME = {30},
      YEAR = {1928},
    NUMBER = {2},
     PAGES = {275--306},
       DOI = {10.2307/1989123},
       URL = {https://doi.org/10.2307/1989123},
}

@article{Alexander1923,
  title = {A Lemma on Systems of Knotted Curves},
  volume = {9},
  ISSN = {1091-6490},
  url = {http://dx.doi.org/10.1073/pnas.9.3.93},
  DOI = {10.1073/pnas.9.3.93},
  number = {3},
  journal = {Proceedings of the National Academy of Sciences},
  publisher = {Proceedings of the National Academy of Sciences},
  author = {Alexander,  J. W.},
  year = {1923},
  month = mar,
  pages = {93–95}
}

@article{Hass1999,
  title = {The computational complexity of knot and link problems},
  volume = {46},
  ISSN = {1557-735X},
  url = {http://dx.doi.org/10.1145/301970.301971},
  DOI = {10.1145/301970.301971},
  number = {2},
  journal = {Journal of the ACM},
  publisher = {Association for Computing Machinery (ACM)},
  author = {Hass,  Joel and Lagarias,  Jeffrey C. and Pippenger,  Nicholas},
  year = {1999},
  month = mar,
  pages = {185–211}
}

@article{Lackenby2021,
  title = {The efficient certification of knottedness and Thurston norm},
  volume = {387},
  ISSN = {0001-8708},
  url = {http://dx.doi.org/10.1016/j.aim.2021.107796},
  DOI = {10.1016/j.aim.2021.107796},
  journal = {Advances in Mathematics},
  publisher = {Elsevier BV},
  author = {Lackenby,  Marc},
  year = {2021},
  month = aug,
  pages = {107796}
}

@misc{burton2014fastbranchingalgorithmunknot,
      title={A fast branching algorithm for unknot recognition with experimental polynomial-time behaviour}, 
      author={Benjamin A. Burton and Melih Ozlen},
      year={2014},
      eprint={1211.1079},
      archivePrefix={arXiv},
      primaryClass={math.GT},
      url={https://arxiv.org/abs/1211.1079}, 
}

@misc{snappy,
  Author = {Culler, M and Dunfield, N M and Weeks, J R},
  Howpublished = {\texttt{http://\allowbreak snappy.\allowbreak computop.\allowbreak org/}},
  Title = {{SnapPy}, a computer program for studying the geometry and topology of 3-manifolds},
  Universe = {Maths},
  Year = {1991--2018}}

@misc{regina,
  Author = {Burton, Benjamin A. and Budney, Ryan and Pettersson, William and others},
  Howpublished = {\texttt{http://\allowbreak regina.\allowbreak sourceforge.\allowbreak net/}},
  Title = {Regina: Software for 3-Manifold Topology and Normal Surface Theory},
  Universe = {Maths},
  Year = {1999--2018}}

@article{burton04-regina,
  Author = {Burton, B A},
  Fjournal = {Experimental Mathematics},
  Journal = {Experiment. Math.},
  Mrclass = {57N10 (57M27)},
  Title = {Introducing {R}egina, the 3-Manifold Topology Software},
  Universe = {Maths},
  Year = 2004}

@InProceedings{burton:LIPIcs.SoCG.2020.25,
  author =  {Burton, Benjamin A.},
  title = {{The Next 350 Million Knots}},
  booktitle = {36th International Symposium on Computational Geometry (SoCG 2020)},
  pages = {25:1--25:17},
  series =  {Leibniz International Proceedings in Informatics (LIPIcs)},
  ISBN =  {978-3-95977-143-6},
  ISSN =  {1868-8969},
  year =  {2020},
  volume =  {164},
  editor =  {Cabello, Sergio and Chen, Danny Z.},
  publisher = {Schloss Dagstuhl -- Leibniz-Zentrum f{\"u}r Informatik},
  address = {Dagstuhl, Germany},
  URL =   {https://drops.dagstuhl.de/entities/document/10.4230/LIPIcs.SoCG.2020.25},
  URN =   {urn:nbn:de:0030-drops-121831},
  doi =   {10.4230/LIPIcs.SoCG.2020.25}
}

@article{MAKOWSKY2003742,
title = {The parametrized complexity of knot polynomials},
journal = {Journal of Computer and System Sciences},
volume = {67},
number = {4},
pages = {742-756},
year = {2003},
note = {Parameterized Computation and Complexity 2003},
issn = {0022-0000},
doi = {https://doi.org/10.1016/S0022-0000(03)00080-1},
url = {https://www.sciencedirect.com/science/article/pii/S0022000003000801},
author = {J.A. Makowsky and J.P. Mariño}
}

@InProceedings{maria:LIPIcs.SoCG.2021.53,
  author =  {Maria, Cl\'{e}ment},
  title = {{Parameterized Complexity of Quantum Knot Invariants}},
  booktitle = {37th International Symposium on Computational Geometry (SoCG 2021)},
  pages = {53:1--53:17},
  series =  {Leibniz International Proceedings in Informatics (LIPIcs)},
  ISBN =  {978-3-95977-184-9},
  ISSN =  {1868-8969},
  year =  {2021},
  volume =  {189},
  editor =  {Buchin, Kevin and Colin de Verdi\`{e}re, \'{E}ric},
  publisher = {Schloss Dagstuhl -- Leibniz-Zentrum f{\"u}r Informatik},
  address = {Dagstuhl, Germany},
  URL =   {https://drops.dagstuhl.de/entities/document/10.4230/LIPIcs.SoCG.2021.53},
  URN =   {urn:nbn:de:0030-drops-138527},
  doi =   {10.4230/LIPIcs.SoCG.2021.53}
}

@InProceedings{burton:LIPIcs.SoCG.2018.18,
  author =  {Burton, Benjamin A.},
  title = {{The HOMFLY-PT Polynomial is Fixed-Parameter Tractable}},
  booktitle = {34th International Symposium on Computational Geometry (SoCG 2018)},
  pages = {18:1--18:14},
  series =  {Leibniz International Proceedings in Informatics (LIPIcs)},
  ISBN =  {978-3-95977-066-8},
  ISSN =  {1868-8969},
  year =  {2018},
  volume =  {99},
  editor =  {Speckmann, Bettina and T\'{o}th, Csaba D.},
  publisher = {Schloss Dagstuhl -- Leibniz-Zentrum f{\"u}r Informatik},
  address = {Dagstuhl, Germany},
  URL =   {https://drops.dagstuhl.de/entities/document/10.4230/LIPIcs.SoCG.2018.18},
  URN =   {urn:nbn:de:0030-drops-87311},
  doi =   {10.4230/LIPIcs.SoCG.2018.18}
}

@article{Kuperberg2015,
  volume = {11},
  ISSN = {1557-2862},
  url = {http://dx.doi.org/10.4086/toc.2015.v011a006},
  DOI = {10.4086/toc.2015.v011a006},
  number = {1},
  journal = {Theory of Computing},
  publisher = {Theory of Computing Exchange},
  author = {Kuperberg,  Greg},
  year = {2015},
  pages = {183–219}
}

@article{Burton2018,
  title = {Algorithms and complexity for {T}uraev–{V}iro invariants},
  volume = {2},
  ISSN = {2367-1734},
  url = {http://dx.doi.org/10.1007/s41468-018-0016-2},
  DOI = {10.1007/s41468-018-0016-2},
  number = {1–2},
  journal = {Journal of Applied and Computational Topology},
  publisher = {Springer Science and Business Media LLC},
  author = {Burton,  Benjamin A. and Maria,  Clément and Spreer,  Jonathan},
  year = {2018},
  month = aug,
  pages = {33–53}
}

@article{Kauffman1990,
  author = {Louis H. Kauffman},
  title = {State models for link polynomials},
  volume = {36},
  number = {1–2},
  journal = {Enseignement Math{\'e}matique},
  year = {1990},
  pages = {1-37}
}

@InProceedings{lunel_et_al:LIPIcs.SoCG.2023.50,
  author =  {Lunel, Corentin and de Mesmay, Arnaud},
  title = {{A Structural Approach to Tree Decompositions of Knots and Spatial Graphs}},
  booktitle = {39th International Symposium on Computational Geometry (SoCG 2023)},
  pages = {50:1--50:16},
  series =  {Leibniz International Proceedings in Informatics (LIPIcs)},
  ISBN =  {978-3-95977-273-0},
  ISSN =  {1868-8969},
  year =  {2023},
  volume =  {258},
  editor =  {Chambers, Erin W. and Gudmundsson, Joachim},
  publisher = {Schloss Dagstuhl -- Leibniz-Zentrum f{\"u}r Informatik},
  address = {Dagstuhl, Germany},
  URL =   {https://drops.dagstuhl.de/entities/document/10.4230/LIPIcs.SoCG.2023.50},
  URN =   {urn:nbn:de:0030-drops-179002},
  doi =   {10.4230/LIPIcs.SoCG.2023.50}
}

@article{DBLP:journals/jocg/SchleimerMPS19,
  author       = {Saul Schleimer and
                  Arnaud de Mesmay and
                  Jessica S. Purcell and
                  Eric Sedgwick},
  title        = {On the tree-width of knot diagrams},
  journal      = {J. Comput. Geom.},
  volume       = {10},
  number       = {1},
  pages        = {164--180},
  year         = {2019},
  url          = {https://doi.org/10.20382/jocg.v10i1a6},
  doi          = {10.20382/JOCG.V10I1A6},
  bibsource    = {dblp computer science bibliography, https://dblp.org}
}

@article{10.1145/3147.3165,
author = {Vitter, Jeffrey S.},
title = {Random sampling with a reservoir},
year = {1985},
issue_date = {March 1985},
publisher = {Association for Computing Machinery},
address = {New York, NY, USA},
volume = {11},
number = {1},
issn = {0098-3500},
url = {https://doi.org/10.1145/3147.3165},
doi = {10.1145/3147.3165},
journal = {ACM Trans. Math. Softw.},
month = mar,
pages = {37–57},
numpages = {21}
}

\appendix
\section{Hecke algebra and knot invariants} \label{app:homflypthecke}

The Hecke algebra is a central object in quantum algebra and representation theory, arising as the endomorphism algebra of $U_q(\mathfrak{gl}_N)$ intertwiners of the $n$-fold tensor product of the vector representation of $U_q(\mathfrak{gl}_N)$, for $N\geq n$. It thus gives a presentation by generators and relations of the endomorphism algebra. By definition, it admits a basis indexed by elements of the symmetric group, but it might be worth mentioning that Kazhdan and Lusztig introduced in the late 70's a different basis~\cite{KL} and proved it to have positive structure constants for the multiplication.

The Hecke algebra arises as a quotient of the braid group by relations that are universal for the Reshetikhin-Turaev knot invariants~\cite{RT}. This includes the Jones polynomial~\cite{Jones} and the Alexander polynomial~\cite{Alexander}, both of which generalize to the HOMFLY-PT polynomial~\cite{HOMFLY,PT}. This latter HOMFLY-PT polynomial can be defined via skein relations as follows:

\begin{definition} \label{def:HomflyPt}
  We define $P$ to be the link invariant that assigns to an oriented framed link an element of $\Z[q^{\pm 1},a^{\pm 1},\frac{a-a^{-1}}{q-q^{-1}}]$ so that:
  \begin{itemize}
  \item     $q^{-1}P\left(\begin{tikzpicture}[anchorbase,scale=.4] \draw [->] (0,0) -- (1,1); \draw (1,0) -- (.6,.4); \draw [->] (.4,.6) -- (0,1); \end{tikzpicture}\right)-qP\left(\begin{tikzpicture}[anchorbase,scale=.4] \draw [->] (1,0) -- (0,1); \draw (0,0) -- (.4,.4); \draw [->] (.6,.6) -- (1,1); \end{tikzpicture}\right)=(q-q^{-1})P\left(\begin{tikzpicture}[anchorbase,scale=.4] \draw [->] (0,0) to [out=60,in=-60] (0,1); \draw [->] (1,0) to [out=120,in=-120] (1,1); \end{tikzpicture}\right)$;
  \item $P\left(\begin{tikzpicture}[anchorbase,scale=.5] \draw (0,0) circle (.5) ; \end{tikzpicture}\right)=\frac{a-a^{-1}}{q-q^{-1}}$;
  \item $P\left(\begin{tikzpicture}[anchorbase,scale=.5] \draw (0,0) to [out=90,in=180] (.5,1.25) to [out=0,in=0] (.5,.75) to [out=180,in=-60] (.2,1); \draw [->] (.1,1.2) to [out=120,in=-90] (0,2);\end{tikzpicture}
  \right)=a^{-1}qP\left(\;
  \begin{tikzpicture}[anchorbase,scale=.5]
    \draw [->] (0,0) -- (0,2);
  \end{tikzpicture}
  \;\right)$.
  \end{itemize}
\end{definition}

Then it is easy to check that this invariant respects the defining relation of the Hecke algebra:
\[
\left(\begin{tikzpicture}[anchorbase,scale=.5] \draw [->] (0,0) -- (1,1); \draw (1,0) -- (.6,.4); \draw [->] (.4,.6) -- (0,1); \end{tikzpicture}\; -\; \begin{tikzpicture}[anchorbase,scale=.5] \draw [->] (0,0) -- (0,1); \draw [->] (1,0) -- (1,1);\end{tikzpicture}\right)\left(\begin{tikzpicture}[anchorbase,scale=.5] \draw [->] (0,0) -- (1,1); \draw (1,0) -- (.6,.4); \draw [->] (.4,.6) -- (0,1); \end{tikzpicture}\; +q^2\; \begin{tikzpicture}[anchorbase,scale=.5] \draw [->] (0,0) -- (0,1); \draw [->] (1,0) -- (1,1);\end{tikzpicture}\right)=0
\]

This explains the relevance of the Hecke algebra in computing knot invariants. 

The following lemma computes the value of stabilized unknots, and is useful for example for the proof of Lemma~\ref{lem:ValueAnnRed}.
\begin{lemma} \label{lem:s_unknot}
  The value of the HOMFLY-PT polynomial of the closure of a cycle (l,1,\dots, l-1) as they appear in Definition~\ref{def:annred} is as follows:
  \[
  P(\hat{\sigma_{l-1}\sigma_{l-2}\cdots \sigma_1})=a^{1-l}q^{l-1}\frac{a-a^{-1}}{q-q^{-1}}.
  \]
\end{lemma}

\begin{proof}
  The third item from Definition~\ref{def:HomflyPt} implies that:
  \[
  P(\hat{\sigma_{l-1}\sigma_{l-2}\cdots \sigma_1})=(a^{-1}q)  P(\hat{\sigma_{l-2}\sigma_{l-2}\cdots \sigma_1})=\cdots =(a^{-1}q)^{l-1}P(\hat{\id_1}),
  \]
  where $\id_1$ stands for the identity braid on one strand. Then the second item implies that $P(\hat{\it_1})=\frac{a-a^{-1}}{q-q^{-1}}$, and the result follows.
\end{proof}

\section{Basic facts about permutations}
We sketch a proof of Lemma~\ref{lem:basic_perm} containing basic facts about encoding permutations. The ideas rely on standard algorithmic techniques.

\begin{proof}{Proof of Lemma~\ref{lem:basic_perm}}
\noindent (i). First, note that turning the index $F_{L(w)}$ into the Lehmer code $L(w) := (L(w)_1, \ldots, L(w)_n)$ is done by decoding the factorial number system encoding. Given the definition,
\[
  (a_1, \ldots, a_n) \rightarrow F_{{\bf a}} = a_1 (n-1)! + a_2 (n-2)! + \cdots a_{n-2} (2!) + a_{n-1} (1!)  \ \ \text{width} \ \ 0 \leq a_i \leq n-i,
\]
one notices that taking successive remainders with Euclidean division extracts the values $a_i$:

\begin{algorithm}[H]
\SetAlgoLined
    $x \leftarrow F_{{\bf a}}$\;
    $a_n \leftarrow 0$\;
    \For{$i \leftarrow n-1$ \KwTo $0$}{
      $a_i \leftarrow \texttt{Remainder}(x / (n-i))$\;
      $x \leftarrow x-a_i$\;
      $X \leftarrow \texttt{Quotient}(x / (n-i))$\;
      }
\end{algorithm}
where \texttt{Remainder} and \texttt{Quotient} compute respectively the remainder and the quotient of the Euclidean division. Assuming these operations take constant time, the decoding above takes $O(n)$ operations.

Now, given a Lehmer code $L(w) := (L(w)_1, \ldots, L(w)_n)$ for a permutation, there is a $O(n \log n)$ algorithm to compute the one-line word presentation $w(1)w(2) \cdots w(n)$ of the permutation.

Maintain the ordered sequence $\{1, \ldots, n\}$. By definition, the first element $w(1)$ is equal to $L(w)_1 + 1$. Iteratively, at step $i$, $w(i)$ is equal to the $(L(w)_i+1)$-th element of the set $\{1, \cdots , n\} \setminus \{w(1), \ldots, w(i-1)\}$. By maintaining the ordered sequence $\{1, \ldots , n\} \setminus \{w(1), \ldots, w(i-1)\}$ in a {\em balanced binary tree}, whose inner nodes maintain the number of leaves of the corresponding subtree, we can erase elements and maintain the balanced property in $O(\log n)$ operations, and access the $k$-th element in $O(\log n)$ operations.

\medskip

\noindent (ii). Uses similarly ideas as $(ii)$. The number of inversions can be computed efficiently by maintaining an ordered list of elements in a balanced binary tree. The factorial number system encoding can be computed efficiently using an adapted version of Horner's method for polynomial evaluation.

\medskip

\noindent (iii). Follows by tracking definitions. Consider $w(1) \cdots w(n) \leq_{\operatorname{lex}} w'(1) \cdots w'(n)$, and let $i$ be the smallest index such that $w(i) \neq w'(i)$ ; in which case we have $w(i) < w'(i)$. Having $w(j)=w('j)$ for $1 \leq j < i$ implies that $L(w)_j = L(w')_j$ for $1 \leq j < i$, and $w(i) < w'(i)$ implies that $L(w)_j < L(w')_j$, by definition. To conclude, we notice that the factorial number system gives enough weight to the $i$th entry to guarantee that $F_{L(w)} < F_{L(w')}$, regardless of the values of $w(i+1) \cdots w(n)$ and $w'(i+1)\cdots w'(n)$.
\end{proof}

\section{Non-faithfulness for $B_4$}
\label{app:resultsB4}
Here we give examples of braids that lie in the kernel of the map $\psi:B_4\rightarrow \He_4$ with coefficients modulo $p$ for $p=2,3,4$.

\begin{example}
Over $\Z/2\Z$, the shorter ones were of the form $s_{\idx}^{-1} \Delta^{-6} b$, where $b$ is a positive braid of Garside length $13$. One such example, with $78$ crossings, is:
\begin{align*}
  b = &\sigma_{1} \sigma_{2} \sigma_{3} \sigma_{3} \sigma_{2} \sigma_{3} \sigma_{2} \sigma_{3} \sigma_{1} \sigma_{3} \sigma_{1} \sigma_{2} \sigma_{2} \sigma_{3} \sigma_{1} \sigma_{2} \sigma_{2} \sigma_{3} \sigma_{1} \sigma_{3} \sigma_{1} \sigma_{2} \sigma_{2} \sigma_{1} \sigma_{2} \sigma_{1} \sigma_{2} \sigma_{3} \sigma_{3} \sigma_{1} \sigma_{2} \sigma_{2} \sigma_{3} \sigma_{1} \sigma_{3} \sigma_{1} \sigma_{3} \sigma_{1} \sigma_{2}
\\
  s_{\idx}^{-1} =& s_{10}^{-1} = \sigma^{-1}_{1} \sigma^{-1}_{3} \sigma^{-1}_{2}.
\end{align*}
\end{example}

\begin{example}
For $\Z/3\Z$, the shortest braids in $B_4$ with trivial Hecke representations had $b$ of Garside length $19$. One of them is $s_{12}^{-1} \cdot \Delta^{-14} \cdot b$, a $172$ crossings braid, with:
\begin{align*}
  b =& \sigma_{2} \sigma_{1} \sigma_{3} \sigma_{2} \sigma_{1} \sigma_{2} \sigma_{3} \sigma_{3} \sigma_{1} \sigma_{2} \sigma_{2} \sigma_{3} \sigma_{3} \sigma_{1} \sigma_{2} \sigma_{2} \sigma_{1} \sigma_{1}
  \sigma_{2} \sigma_{3} \sigma_{3} \sigma_{2} \sigma_{3} \sigma_{3} \sigma_{2} \sigma_{3} \sigma_{1} \sigma_{3} \sigma_{1} \sigma_{2} \sigma_{2} \sigma_{1} \sigma_{2} \sigma_{3} \sigma_{3} \sigma_{1} 
  \\
  &\sigma_{2} \sigma_{2} \sigma_{2} \sigma_{1} \sigma_{3} \sigma_{1} \sigma_{2} \sigma_{2} \sigma_{2} \sigma_{3} \sigma_{1} \sigma_{3} \sigma_{1} \sigma_{2} \sigma_{2} \sigma_{1} \sigma_{2} \sigma_{3}
  \sigma_{3} \sigma_{1} \sigma_{2} \sigma_{2} \sigma_{3} \sigma_{3} \sigma_{2} \sigma_{3} \sigma_{1} \sigma_{3} \sigma_{1} \sigma_{2} \sigma_{2} \sigma_{3} \sigma_{3} \sigma_{1} \sigma_{2} \sigma_{2}
  \\
  &\sigma_{3} \sigma_{1} \sigma_{3} \sigma_{2} \sigma_{1} \sigma_{2} \sigma_{3} \sigma_{3} \sigma_{1} \sigma_{2} \sigma_{2} \sigma_{3} \sigma_{1} \sigma_{2}
  \\
  s_{12}^{-1} = &\sigma^{-1}_{2} \sigma^{-1}_{1}.
\end{align*}
\end{example}

\begin{example}
Over $\Z/4\Z$, the shortest braids had $b$ of Garside length $29$. One of them, with a total of $184$ crossings, is $s_{7}^{-1} \cdot \Delta^{-15} \cdot b$ with:
\begin{align*}
  b = &\sigma_{2} \sigma_{3} \sigma_{1} \sigma_{3} \sigma_{1} \sigma_{2} \sigma_{2} \sigma_{3} \sigma_{3} \sigma_{2} \sigma_{1} \sigma_{2} \sigma_{1} \sigma_{2} \sigma_{2} \sigma_{1} \sigma_{2} \sigma_{3}
  \sigma_{3} \sigma_{1} \sigma_{2} \sigma_{2} \sigma_{3} \sigma_{1} \sigma_{3} \sigma_{1} \sigma_{3} \sigma_{1} \sigma_{2} \sigma_{2} \sigma_{3} \sigma_{1} \sigma_{3} \sigma_{1} \sigma_{2} \sigma_{2} \sigma_{2}
  \\
  &\sigma_{3} \sigma_{1} \sigma_{3} \sigma_{1} \sigma_{2} \sigma_{2} \sigma_{3} \sigma_{1} \sigma_{3} \sigma_{1} \sigma_{3} \sigma_{1} \sigma_{2} \sigma_{2} \sigma_{3} \sigma_{3} \sigma_{2} \sigma_{1} \sigma_{2}
  \sigma_{1} \sigma_{2} \sigma_{2} \sigma_{1} \sigma_{2} \sigma_{3} \sigma_{3} \sigma_{1} \sigma_{2} \sigma_{2} \sigma_{3} \sigma_{1} \sigma_{2} \sigma_{2} \sigma_{3} \sigma_{1} \sigma_{3} \sigma_{1} 
  \\
  &\sigma_{2} \sigma_{2} \sigma_{3} \sigma_{1} \sigma_{3} \sigma_{2} \sigma_{1} \sigma_{2} \sigma_{3} \sigma_{3} \sigma_{1} \sigma_{2} \sigma_{2} \sigma_{3} \sigma_{1} \sigma_{3} \sigma_{1} \sigma_{2}
\\
  s_7^{-1} = &\sigma^{-1}_{1} \sigma^{-1}_{3}.
\end{align*}
\end{example}

\end{document}